\documentclass[sigplan,screen]{acmart}

\copyrightyear{2020}
\acmYear{2020}
\setcopyright{acmcopyright}
\acmConference[LICS '20]{Proceedings of the 35th Annual ACM/IEEE Symposium on Logic
	in Computer Science (LICS)}{July 8--11, 2020}{Saarbr\"ucken, Germany}
\acmBooktitle{Proceedings of the 35th Annual ACM/IEEE Symposium on Logic in
	Computer Science, July 8--11, 2020}
\acmPrice{15.00}
\acmDOI{10.1145/3373718.3394753}
\acmISBN{978-1-4503-7104-9/20/07}
\startPage{1}

\usepackage{booktabs}   
\usepackage{subcaption} 

\usepackage{mathtools}
\usepackage{mathrsfs}

\theoremstyle{definition}
\newtheorem{theorem}{Theorem}[section]
\newtheorem{lemma}[theorem]{Lemma}
\newtheorem{prop}[theorem]{Proposition}

\newtheorem{corollary}[theorem]{Corollary}
\newtheorem{definition}[theorem]{Definition}
\newtheorem{example}[theorem]{Example}

\newcommand{\nn}{\mathbb{N}}
\newcommand{\maj}{\sqsubseteq^{\mathtt{maj}}}
\newcommand{\smy}{\sqsubseteq^{\mathtt{min}}}
\newcommand{\lemaj}{\le^{\mathtt{maj}}}
\newcommand{\lesmy}{\le^{\mathtt{min}}}
\newcommand{\pow}{\mathbb{P}_{\text{f}}}
\newcommand{\refl}{\mathcal{R}}
\newcommand{\powi}{\mathbb{P}}

\newcommand{\card}{\mathtt{card}}
\newcommand{\comp}{\mathtt{comp}}
\newcommand{\proj}{\mathtt{proj}}
\newcommand{\CNF}{\text{CNF}}
\newcommand{\fastf}{\mathscr{F}}
\newcommand{\hri}{\hookrightarrow}
\newcommand{\xhri}{\xhookrightarrow}

\begin{document}
	
	\title{Complexity of controlled bad sequences over finite sets of $\mathbb{N}^d$}         

	
	\author{A. R. Balasubramanian}
	\authornote{Work done when the author was an intern at LSV, ENS-Saclay Paris}          
	\orcid{0000-0002-7258-5445}             
	\affiliation{
		\institution{Technical University of Munich}            
		\city{Munich}
		\country{Germany}                    
	}
	\email{bala.ayikudi@tum.de}          

	\begin{abstract}
		We provide upper and lower bounds for the length of controlled bad sequences over the majoring and the minoring orderings of finite sets of $\mathbb{N}^d$. The results are obtained by bounding the length of such sequences by functions from the Cichon hierarchy. This allows us to translate these results to bounds over the fast-growing complexity classes. 
		
		The obtained bounds are proven to be tight for the majoring ordering, which solves a problem left open by Abriola, Figueira and Senno (Theor. Comp. Sci, Vol. 603). Finally, we use the results on controlled bad sequences to prove upper bounds for the emptiness problem of some classes of automata. 
	\end{abstract}



	\maketitle

	\section{Introduction}
	
	A well-quasi order (wqo) over a set $A$ is a reflexive and transitive relation $\le_A$ such that every infinite sequence $x_0,x_1,x_2,\dots$ over $A$ has an increasing pair $x_i \le_A x_j$ with $i < j$. 
	A normed wqo (nwqo) is a wqo $(A,\le_A)$ which has a \emph{norm} function $|\cdot| : A \to \nn$ such that the pre-image of $n$ under $|\cdot|$ is finite for every $n$. 	
	A sequence over a wqo is called a bad sequence if it contains no increasing pair. Hence, all bad sequences over a well-quasi order are necessarily finite.
	
	Well-quasi orders are an important tool in logic, combinatorics and computer science as evidenced by their applications in term-rewriting systems \cite{Rewrite}, algorithms \cite{Algorithms, Embedding} and verification of infinite state systems \cite{wqts,wqts1,wsts}. Indeed, well-quasi orders form the backbone for the ubiquitous \emph{well-structured transition systems} (wsts) \cite{wsts,wqts}, whose coverability problem is shown to be decidable thanks to well-quasi orders. 
	
	In recent years, significant effort has been put in to understand the 
	complexity of the coverability procedure for various well-structured transition systems (See \cite{LICS,ICALP,Linearise,DataNets,LCS} and also \cite{BeyondElem} for a catalogue of many problems). 
	The key idea behind proving upper bounds for the coverability algorithm is the following: For a given class of wsts, the running time of the  coverability procedure for that class can be bounded by the length of \emph{controlled bad sequences} of the underlying normed well-quasi order (See definition \ref{def:controlledseq} for a formal definition of controlled bad sequences). 
	Intuitively, for a function $g$ and a number $n$, a sequence $x_0,x_1,x_2,\dots,x_l$ is called a $(g,n)$-controlled bad sequence if $|x_0| \le n, |x_1| \le g(n), |x_2| \le g(g(n))$ and so on. A simple application of Konig's lemma can be used to show that for every $n$, there is a $(g,n)$-controlled bad sequence of maximum length. Hence, for every function $g$ we can define a \emph{length function} which maps a number $n$ to the length of the longest $(g,n)$-controlled bad sequence. 
	
	The main observation made in \cite{LICS,ICALP,BeyondElem,Linearise,LCS,DataNets} is that, for various classes of well-structured systems, an upper bound on the running time of the coverability procedure could be obtained by bounding the length function of some specific $g$ over the underlying wqo of that class. Motivated by this, upper bounds on the length of controlled bad sequences have been obtained for various well-quasi orders: The product ordering over $\nn^d$ (\cite{LICS}), the lexicographic ordering over $\nn^d$ (\cite{Linearise}),
	the multiset ordering over multisets of $\nn^d$ (\cite{Linearise}), the subword ordering over words \cite{ICALP} and the linear ordering over ordinals \cite{Ordinals}, to name a few. Using these results, time bounds have been established for the following problems (See \cite{BeyondElem} for a more detailed overview): coverability of lossy counter machines, coverability and termination of lossy channel systems, coverability of unordered data nets, emptiness of alternating 1-register and 1-clock automata, the regular Post embedding problem, conjunctive relevant implication and 1-dimensional VASS universality. The present work is a contribution in this field of inquiry.
	
	\textbf{Our contributions: } In this paper, we prove lower and upper bounds on the length of controlled bad sequences for the \emph{majoring} and \emph{minoring} ordering (See definition \ref{def:majandmin}) over the collection of all finite sets of $\nn^d$ (hereafter denoted by $\pow(\nn^d)$). Both orderings have been used to prove the decidability of the emptiness problem for some classes of automata in \cite{ATRA} and \cite{BUDA}. Our main results are the following: 
	\begin{itemize}
		\item We show that if the function $g$ is primitive recursive, then the length function of $g$ for the majoring ordering over $\pow(\nn^d)$ is bounded by a function in the complexity class $\fastf_{\omega^{d-1}}$ (For a definition of $\fastf_{\omega^{d-1}}$, see section~\ref{section:complexity}). We also show that the length function of $g$ for the minoring ordering over $\pow(\nn^d)$ is bounded by a function in $\fastf_{\omega^{d-1} \cdot 2^d}$	
		\item To complement the upper bounds, we also provide lower bounds on the length functions. We prove the existence of a primitive recursive (in fact, polynomial) function $g$ such that the length function of $g$ over the majoring ordering is bounded from below by a function in $\fastf_{\omega^{d-1}}$. A similar result is also obtained for the minoring ordering.
		\item We use the upper bounds on the length functions to provide upper bounds on the running time for the emptiness problem of some classes of automata operating on trees.
	\end{itemize}
	
	\textbf{Related work: } Length functions for the majoring ordering over $\pow(\nn^d)$ was considered in \cite{Linearise}, where the authors proved an upper bound of $\fastf_{\omega^d}$. However no lower bound was provided and the authors left open the question of the tightness of their bound. Our results (theorems \ref{theorem:complowboundmaj}, \ref{theorem:compuppboundmaj}) show that their bound is not optimal and gives tight upper and lower bounds. Some results concerning the minoring ordering were presented in \cite{Minoring}, but no bounds on the length function were proven. To the best of our knowledge, we provide the first upper bounds for length functions of the minoring ordering over $\pow(\nn^d)$.

	\textbf{Our techniques: } Various results regarding length functions have been proven using the notion of a \emph{reflection} from one normed wqo to another (See definition 3.3 of \cite{ICALP} or definition \ref{def:polyrefl}). A reflection is a map from one nwqo $A$ to another nwqo $B$, which satisfies some properties on the order and norm. If a reflection exists from $A$ to $B$ it can be easily proven that the length function of $g$ over $A$ is less than the length function of $g$ over $B$. However, it turns out that reflections are not sufficient for our purposes. To this end, we define a generalization of reflections called \emph{polynomial reflections} (See definition \ref{def:polyrefl}). We show that if a polynomial reflection exists from $A$ to $B$, then bounds on the length function for $g$ over $A$ can be easily transferred to bounds on the length function for $h$ over $B$, where $h$ is a function obtained by composing a polynomial with $g$. 
	
	We then show that there exists a polynomial reflection from the set of ordinals less than $\omega^{\omega^{d-1}}$ (with the usual ordinal ordering) to $\pow(\nn^d)$ with the majoring ordering (Lemma \ref{lemma:lowmaj}). This enables us to establish a lower bound for the majoring ordering in terms of lower bounds for the order on ordinals, which are already known (\cite{Ordinals}).
	
	The upper bound for the majoring ordering is proved by following the framework established by Schmitz and Schnoebelen in a series of papers (\cite{ICALP,LICS,Ordinals}), which we briefly describe here. It is well known that using the descent equation, the length function for a nwqo can be expressed inductively by length functions over its ``residuals''. However, the residuals of a nwqo can become extremely complex to derive any useful bounds for the length function. To overcome this, we associate an ordinal to each residual (called the order type) and a non-trivial ``derivative'' operator for each ordinal. We then show that in the descent equation, we can replace the residuals of a nwqo with the derivative operator of the order type of that nwqo, which are much more amenable to analysis. Once this is carried out, we exploit some properties of the derivative operator along with some facts about the Cichon hierarchy and ordinal ordering to establish bounds on the length function.
	
	The lower bound for the minoring ordering is established by giving a simple polynomial reflection from the majoring ordering to the minoring ordering (Lemma \ref{lemma:lowmin}). By using the lower bounds proved for the majoring ordering, we can infer lower bounds for the minoring ordering. Finally, the upper bound for the minoring ordering is established by giving a non-trivial polynomial reflection from the minoring ordering to a \emph{cartesian product} of various majoring orderings (Lemma \ref{lemma:imporlemma}). The intuition behind the reflection is discussed in detail in section \ref{section:uppmin}.
	
	\textit{Outline of the paper: } We recall basic notions of wqos, ordinals and sub-recursive hierarchies in section \ref{section:prelim}. In sections \ref{section:lowmaj} and \ref{section:uppmaj} we prove lower and upper bounds for the majoring ordering in terms of functions from the Cichon hierarchy. Similar results are proved in sections \ref{section:lowmin} and \ref{section:uppmin} for the minoring ordering. We give a classification of these bounds in the fast-growing hierarchy in section \ref{section:complexity}. Finally, we conclude with providing some applications of our results in \ref{section:applications}.
	

	\section{Preliminaries} \label{section:prelim}
	
	We recall some basic facts about well-quasi orders (see \cite{Basics}). 
	A \emph{quasi ordering} (qo) over a set $A$ is a relation $\le$ such that $\le$ is reflexive and transitive. We write $x < y$ if $x \le y$ and $y \nleq x$. 
	We also say $x \equiv y$ if $x \le y $ and $y \le x$.
	A \emph{well-quasi ordering} (wqo) over a set $A$ is a qo $\le$ such that for every infinite sequence $x_0,x_1,x_2,\dots,$ there exists $i < j$ such that $x_i \le x_j$. A \emph{norm function} over a set $A$ is a function $|\cdot| : A \to \nn$ such that for every $n \in \nn$ the set $\{x \in A : |x| < n\}$ is finite.
	
	\begin{definition}
		A \emph{normed wqo} is a wqo $(A,\le_A,|\cdot|_A)$ equipped with a norm function $|\cdot|_A$. 
	\end{definition} 
	
	The set $A$ will be called as the domain of the nwqo. 
	If $(A,\le_A,|\cdot|_A)$ is a nwqo and $S \subseteq A$ then the nwqo induced by $S$ is the nwqo $(S,\le_S,|\cdot|_S)$ where $\le_S$ and $|\cdot|_S$ are the restrictions of $\le_A$ and $|\cdot|_A$ on $S$ respectively. We use the notation $A_{\le n}$ to define the set $\{x \in A : |x| \le n\}$. Whenever the order $\le_A$ or the norm $|\cdot|_A$ is clear from the context, we will drop those and just refer to the nwqo by the domain $A$.
	
	\begin{example} (\emph{Some basic nwqos}) : 
		The set of natural numbers with the usual ordering and the identity norm $(\nn,\le,\texttt{id})$ is clearly seen to be a nwqo. 
		Another nwqo is any finite set $\{a_0,a_1,\dots,a_{k-1}\}$ such that distinct letters are unordered and $|a_i| = 0$ for every $i$. We will denote this nwqo by $\Gamma_k$. Notice that $\Gamma_0$ is the empty nwqo.
	\end{example}
	
	Given two nwqos $A$ and $B$ we write $A \equiv B$ when $A$ and $B$ are \emph{isomorphic} structures. In particular the norm functions must be preserved by the isomorphism.
	
	\subsection*{Good, bad and controlled sequences}
	\begin{definition} \label{def:goodandbad}
		A sequence $x_0,x_1,\dots$ over a qo $(A,\le_A)$ is called \emph{good} if there exists $i < j$ such that $x_i \le_A x_j$. A sequence which is not good is called \emph{bad}. Notice that every bad sequence in a wqo is necessarily finite.
	\end{definition}
	
	\begin{definition} \label{def:controlledseq}
		A \emph{control function} is a mapping $g : \nn \to \nn$. For an $n \in \nn$, a sequence $x_0,x_1,\dots$ over a nwqo $A$ is $(g,n)$-controlled if
		\begin{equation*}
		\forall i \in \nn, \ |x_i|_A \le g^i(n) = \overbrace{g(g(\dots(g}^{i \text{ times}}(n))))
		\end{equation*}
	\end{definition}
	
	By a straightforward application of Konig's lemma, we have the following proposition: (See proposition 2.5 of \cite{ICALP})
	
	\begin{prop} 
		Let $A$ be a nwqo and let $g$ be a control function. For every $n \in \nn$, there exists a finite maximum length $L \in \nn$ for $(g,n)$-controlled bad sequences over $A$.
	\end{prop}
	
	Therefore the above proposition lets us define a function $L_{A,g} : \nn \to \nn$ which for every $n \in \nn$, assigns the maximum length of a $(g,n)$-controlled bad sequence over $A$. 
	We will call this \emph{the length function of $A$ and $g$}. 
	From now on, we assume that $g$ is a \emph{strictly increasing inflationary} function 
	(Here inflationary means that $g(n) \ge n$ for all $n \in \nn$).
	
	
	\subsection*{Descent equation}
	
	We can express the length function by induction over nwqos. To do this we need the notion of \emph{residuals}.
	
	\begin{definition} \label{def:residuals}
		Let $A$ be a nwqo and $x \in A$. The \emph{residual} $A/x$ is the nwqo induced by the subset $A/x := \{y \in A : x \nleq_A y\}$
	\end{definition}
	We have the following proposition: (See proposition 2.8 of \cite{ICALP})
	\begin{prop} \label{prop:descenteqn}
		\[L_{A,g}(n) = \max_{x \in A_{\le n}} \{1 + L_{A/x,g}(g(n))\}\]
	\end{prop}
	
	This equation is called the \emph{descent equation}. The descent equation implies that unraveling the length function inductively gives us a way of computing it. If $A \supsetneq A/x_0 \supsetneq A/x_0/x_1 \supsetneq \dots$, it follows that $x_0,x_1,\dots$ is a bad sequence
	and so the inductive unraveling of proposition \ref{prop:descenteqn} is well founded.
	
	\subsection{Constructing Normed Wqo's}
	
	In this section, we will see how to construct ``complex'' nwqos in terms of more simpler nwqos. The constructions we use in this paper are disjoint sums, cartesian products and finite powersets. 
	
	\begin{definition}(Disjoint sum and cartesian product) \label{def:sumandproduct}
		Let $A_1$ and $A_2$ be two nwqos. The \emph{disjoint sum} $A_1+A_2$ is the nwqo given by
		\begin{align*}
			A_1 + A_2 &:= \{(i,x) : i \in \{1,2\} \text{ and } x \in A_i\}\\
			(i,x) \le_{A_1 + A_2} (j,y) &\Leftrightarrow i = j \text{ and } x \le_{A_i} y\\
			|(i,x)|_{A_1 + A_2} &:= |x|_{A_i}	
		\end{align*}
		
		The cartesian product $A_1 \times A_2$ is the nwqo given by
		\begin{align*}
			A_1 \times A_2 &:= \{(x_1,x_2) : x_1 \in A_1, x_2 \in A_2\}\\
			(x_1,x_2) \le_{A_1 \times A_2} (y_1,y_2) &\Leftrightarrow x_1 \le_{A_1} y_1 \text{ and } x_2 \le_{A_2} y_2\\
			|(x_1,x_2)|_{A_1 \times A_2} &:= \max(|x_1|_{A_1},|x_2|_{A_2})	
		\end{align*}
		
		It is well known that both $A_1 + A_2$ and $A_1 \times A_2$ are nwqos when $A_1$ and $A_2$ are. Of special interest to us is the cartesian product $(\nn^d,\le_{\nn^d},|\cdot|_{\nn^d})$ which is obtained by taking cartesian product of $(\nn,\le,\texttt{id})$ with itself $d$ times. 
		From now on, whenever we refer to the underlying order of $\nn^d$, we will always mean this cartesian product ordering.
	\end{definition}
	
	\begin{definition}(Majoring and minoring orderings) \label{def:majandmin}
		Let $A$ be a nwqo. We construct two nwqos whose domain will be the set of all finite subsets of $A$, which we denote by $\pow(A)$.
		The first is called the \emph{majoring ordering} and is defined as
		\begin{align*}
			\pow(A) &:= \{X : X \subseteq A \text{ and } X \text{ is finite}\}\\
			X \maj_{\pow(A)} Y &\Leftrightarrow \forall x \in X, \exists y \in Y \text{ such that } x \le_A y\\
			|X|_{\pow(A)} &:= \max(\{|x|_A : x \in X\},\card(X))
		\end{align*}
		Here $\card(X)$ denotes the cardinality of the set $X$. 
		
		The second is called the \emph{minoring ordering} and it has the same domain and the norm as that of the majoring ordering. The difference lies in the ordering, which is given by
		$$X \smy_{\pow(A)} Y \Leftrightarrow \forall y \in Y, \exists x \in X \text{ such that } x \le_A y$$
		
		The fact that $(\pow(A),\maj_{\pow(A)},|\cdot|_{\pow(A)})$ is a nwqo easily follows from Higman's lemma (\cite{Higman}). However $(\pow(A),\smy_{\pow(A)},|\cdot|_{\pow(A)})$ is \emph{not necessarily} a nwqo whenever $A$ is (\cite{Minoring}). But, it is known that $(\pow(\nn^d),\smy_{\pow(\nn^d)},|\cdot|_{\pow(\nn^d)})$ is a nwqo for every $d$ (See \cite{Minoring}). Whenever there is no confusion, we drop the $\pow(A)$ as a subscript and refer to the majoring (resp. minoring ) nwqo as $(\pow(A),\maj)$ (resp. $(\pow(A),\smy)$).
		
	\end{definition}
	
	The results that we prove in this paper will only concern the nwqos
	$(\pow(\nn^d),\maj,|\cdot|_{\pow(\nn^d)})$ and $(\pow(\nn^d),\smy,|\cdot|_{\pow(\nn^d)})$. However, for the purposes of our proofs, we also need the following wqos which can be seen as extensions of the majoring and minoring ordering to the set of \emph{all} subsets of a wqo.
	
	\begin{definition} (Arbitrary subsets) \label{def:arbsubsets}
		Let $(A,\le_A)$ be a wqo and let $\powi(A)$ denote the set of all subsets (finite and infinite) of $A$. Let $X,Y \in \powi(A)$. We define,
		\begin{align*}
			X \maj Y &\iff \forall x \in X, \exists y \in Y \text{ such that } x \le_A y\\
			X \smy Y &\iff \forall y \in Y, \exists x \in X \text{ such that } x \le_A y
		\end{align*}
		
		Further given $X \in \powi(A)$ define,
		\begin{itemize}
			\item $\min(X) := \{x \in X :  \forall x' \in X, \ [x' \nleq_{A} x \implies x \equiv x']\}$
			\item $\uparrow X = \{a : \exists x \in X, x \le_{A} a\}$
			\item $\downarrow X = \{a : \exists x \in X, a \le_{A} x\}$
		\end{itemize}
	\end{definition}
	
	Notice that if $\le_A$ is also guaranteed to be antisymmetric,
	then $\min(X)$ is always a finite set, irrespective of whether $X$ is finite or infinite. 
	Also, observe that we do \emph{not} endow $\powi(A)$ with a norm.

	\begin{prop} \label{prop:funfacts}
		Let $X,Y \in \powi(\nn^d)$. The following facts are known about $(\powi(\nn^d),\maj)$ and 
		$(\powi(\nn^d),\smy)$ (see \cite{Minoring}):
		\begin{enumerate}
			\item The ordering $(\powi(\nn^d),\maj)$ is a wqo
			\item The ordering $(\powi(\nn^d),\smy)$ is a wqo
			\item $X \maj Y \iff \downarrow X \maj \downarrow Y$
			\item $X \smy Y \iff \uparrow X \smy \uparrow Y$
			\item $X \maj Y \iff \nn^d \setminus \downarrow X \smy \nn^d \setminus \downarrow Y$
			\item $X \smy Y \iff \nn^d \setminus \uparrow X \maj \nn^d \setminus \uparrow Y$
			\item $X \smy Y \iff \min(X) \smy \min(Y)$
		\end{enumerate}
	\end{prop}

	\subsection*{Reflections}
	
	A major tool to prove lower and upper bounds on the length of controlled bad sequences is the notion of a \emph{normed reflection} (See definition 3.3 of \cite{ICALP}). However, for our purposes we require the following notion of a \emph{polynomial normed reflection}.
	
	\begin{definition} \label{def:polyrefl}
		A \emph{polynomial nwqo reflection} is a mapping $r : A \to B$ such that there exists a polynomial $q : \nn \to \nn$ and
		$$\forall x,y \in A : r(x) \le_B r(y) \text{ implies } x \le_A y$$
		$$\forall x \in A : |r(x)|_B \le q(|x|_A)$$
		If these conditions are satisfied then we say that $r$ is a polynomial nwqo reflection with polynomial $q$ and denote it by $r : A \xhri{q} B$. If the polynomial $q$ is the identity function, we call it a \emph{nwqo reflection} and denote it by $r : A \hri B$.
	\end{definition}
	
	It is easy to see that if $r: A \xhri{q} B$ and $r' : B \xhri{q'} C$ are polynomial nwqo reflections, then $r' \circ r : A \xhri{q' \circ q} C$ is also a polynomial nwqo reflection.
	Further, reflections are also a \emph{precongruence} with respect to disjoint sums and cartesian products, i.e.,
	\begin{prop} \label{prop:precongruence} (See \ref{proof:precongruence}) 
		Suppose $r: A \xhri{q} B$ and $r': A' \xhri{q'} B'$ are polynomial nwqo reflections. Then there exists functions $s$ and $p$ such that $s: A+A' \xhri {q+q'} B+B$ and $p: A \times A' \xhri {q+q'} B \times B'$.
	\end{prop}
	
	We have the following important result regarding polynomial nwqo reflections.
	\begin{prop} \label{prop:polyrefl} (See \ref{proof:polyrefl}) 
		Let $r: A \xhri{p} B$ be a polynomial nwqo reflection. Then 
		$L_{A,g}(n) \le L_{B,(q \circ g)}(q(n))$ for some polynomial $q$. Further if $p$ is increasing and inflationary, then it suffices to take $q = p$.
	\end{prop}

	\subsection{Ordinals and subrecursive hierarchies}
	
	Since all our results will be phrased in terms of functions in the Cichon hierarchy, we recall basic facts about ordinals and subrecursive hierarchies in this section.
	
	\subsection*{Ordinal terms} \label{section:ordinalterms}
	For basic notions about ordinals and its ordering, we refer the reader to \cite{Ordinals}. We will use Greek letters $\alpha,\beta,\dots$ to denote ordinals and $\le$ to denote the ordering on ordinals. We will always use $\lambda$ to denote limit ordinals.
	
	An ordinal $\alpha$ has the general form (also called the \emph{Cantor Normal Form}) $\alpha = \omega^{\beta_1} + \omega^{\beta_2} + \dots + \omega^{\beta_m}$ where $\beta_1,\dots,\beta_m$ are ordinals such that
	$\beta_1 \ge \beta_2 \ge \dots \ge \beta_m$. 
	For an ordinal $\alpha$, we let CNF$(\alpha)$ denote the set of all ordinals strictly less than $\alpha$. For the purposes of this paper, we will restrict ourselves to ordinals in CNF$(\epsilon_0)$ (where $\epsilon_0$ is the supremum of $\omega,\omega^{\omega},\omega^{\omega^{\omega}},\cdots$) 
	
	For $c \in \nn$, let $\omega^{\beta} \cdot c$ denote $\overbrace{\omega^{\beta} + \dots + \omega^{\beta}}^{c \text{ times }}$. We sometimes write ordinals in a \emph{strict form } as $\alpha = \omega^{\beta_1} \cdot c_1 + \omega^{\beta_2} \cdot c_2 + \dots + \omega^{\beta_m} \cdot c_m$ where $\beta_1 > \beta_2 > \dots > \beta_m$ and 
	the \emph{coefficients} $c_i$ must be strictly bigger than $0$. Using the strict form, we define a norm $N$ on CNF$(\epsilon_0)$ as follows: if $\alpha = \omega^{\beta_1} \cdot c_1 + \omega^{\beta_2} \cdot c_2 + \dots + \omega^{\beta_m} \cdot c_m$ in the strict form then $N\alpha = \max\{c_1,\dots,c_m,N\beta_1,\dots,N\beta_m\}$. It is not very hard to notice that for every $\alpha < \epsilon_0$, the set $\CNF(\alpha)_{\le n}$ is always finite for any $n$. Hence for every $\alpha < \epsilon_0$, we have a nwqo $($CNF$(\alpha),\le,N)$.
	
	We finish this sub-section with the definitions of \emph{natural sum} ($\oplus$) and \emph{natural product} ($\otimes$) for ordinals in $\CNF(\epsilon_0)$: 
	\begin{align*}
		\sum_{i=1}^m \omega^{\beta_i} \oplus \sum_{j=1}^n \omega^{\beta'_j} &:= \sum_{k=1}^{m+n} \omega^{\gamma_k}\\
		\sum_{i=1}^m \omega^{\beta_i} \otimes \sum_{j=1}^n \omega^{\beta'_j} &:=
		\bigoplus_{i=1}^m \bigoplus_{j=1}^n \omega^{\beta_i \oplus \beta'_j}	
	\end{align*}
	where $\gamma_1 \ge \gamma_2 \dots \ge \gamma_{m+n}$ is a rearrangement of $\beta_1,\dots,\beta_m,$ $\beta'_1,\dots,\beta'_n$.

	As mentioned before, all our results will be obtained by providing reflections to and from the ordinal ordering. Hence, it is important to understand how ``fast'' the length of controlled bad sequences in the ordinal ordering can grow. For this purpose, we introduce sub-recursive hierarchies. 
	
	\subsection*{Sub-recursive hierarchies}
	
	For the purposes of describing the length of controlled bad sequences over the ordinal ordering, the hierarchies of Hardy and Cichon are sufficient \cite{Cichon}. However, before we introduce them we need some preliminary definitions.
	
	A \emph{fundamental sequence} for a \emph{limit ordinal} $\lambda$ is a sequence $(\lambda(x))_{x < \omega}$ with supremum $\lambda$, which we fix to be,
	\begin{equation*}
	(\gamma + \omega^{\beta+1})(x) := \gamma + \omega^{\beta} \cdot (x+1), \qquad 
	(\gamma + \omega^{\lambda})(x) := \gamma + \omega^{\lambda(x)}
	\end{equation*}
	
	The \emph{predecessor} $P_x$ of an ordinal $\alpha > 0$ at $x \in \nn$ is given by
	\begin{equation*}
	P_x(\alpha + 1) := \alpha, \qquad P_x(\lambda) := P_x(\lambda(x))
	\end{equation*}
	
	Let $h : \nn \to \nn$ be a function. The \emph{Hardy hierarchy} for the function $h$ is given by 
	$(h^{\alpha})_{\alpha < \epsilon_0}$ where
	\begin{equation*}
	h^0(x) := x, \qquad h^{\alpha}(x) := h^{P_x(\alpha)}(h(x))
	\end{equation*}
	
	and the \emph{Cichon hierarchy} $(h_{\alpha})_{\alpha < \epsilon_0}$ is defined as
	\begin{equation*}
	h_0(x) := 0, \qquad h_{\alpha}(x) := 1 + h_{P_x(\alpha)}(h(x))
	\end{equation*}
	
	We also define another hierarchy called the \emph{fast growing hierarchy}
	as follows:
	$$f_{h,0}(x) = h(x), \qquad f_{h,\alpha+1}(x) = f^{x+1}_{h,\alpha}(x), \qquad f_{h,\lambda}(x) = f_{h,\lambda_x}(x)$$
	
	Here $f^i_{h,\alpha}$ denotes $i$-fold composition of $f_{h,\alpha}$ with itself. \\
	
	Let $L_{\alpha,g}(n)$ denote the the length of the longest $(g,n)$-controlled bad sequence in CNF$(\alpha)$.
	The following theorem states that, for large enough $n$,
	the length function $L_{\alpha,g}$ and the function $g_{\alpha}$
	in the Cichon hierarchy coincide.
	
	\begin{theorem} (Theorem 3.3 of \cite{Ordinals}) \label{theorem:ordlft}
		Let $\alpha < \epsilon_0$ and $n \ge N\alpha$. Then $L_{\alpha,g}(n) = g_{\alpha}(n)$.
	\end{theorem}

	\section{Lower bound for majoring ordering} \label{section:lowmaj}
	
	In this section we prove a lower bound for length functions over $(\pow(\nn^d),\maj,|\cdot|_{\pow(\nn^d)})$. 
	The lower bound is presented in terms of functions over the
	Cichon hierarchy.
	
	The following lemma follows an unpublished idea of Abriola, Schmitz and Schnoebelen, which has been adapted to controlled bad sequences here.
	
	\begin{lemma} \label{lemma:lowmaj}
		There exists a poly. nwqo reflection
		$$\refl : (\CNF(\omega^{\omega^{d-1}}),\le,N) \xhri{\varphi} (\pow(\nn^d),\maj,|\cdot|_{\pow(\nn^d)})$$ where $\varphi(x) = x(x+1)^d$.
	\end{lemma}
	
	\begin{proof} 
		We decompose the proof into three parts. As a first step, we define the map $\refl$ from $\CNF(\omega^{\omega^{d-1}})$ to $\pow(\nn^d)$. In the second step, we show that $\refl(\gamma) \maj \refl(\zeta) \implies \gamma \le \zeta$. In the third step, we show that $|\refl(\gamma)|_{\pow(\nn^d)} \le \varphi(N\gamma)$ where $N$ is the norm defined on ordinals in section \ref{section:ordinalterms}. 
		
		\subsubsection*{First step}
		
		Let $\gamma \in \CNF(\omega^{\omega^{d-1}})$ such that the Cantor normal form of $\gamma$ is 
		$\omega^{\beta_1} + \omega^{\beta_2} + \dots + \omega^{\beta_l}$.
		Notice that each $\beta_i \in \CNF(\omega^{d-1})$ and hence can be written as $\beta_i = \omega^{d-2} \cdot c_{(i,d-2)} + \omega^{d-3} \cdot c_{(i,d-3)} + \dots + \omega^{0} \cdot c_{(i,0)}$ where the coefficients $c_{i,j}$ \emph{can} be 0. The map $\refl$ is then defined on $\gamma$ as $$\refl(\gamma) := \{(i,c_{(i,0)},c_{(i,1)},\dots,c_{(i,d-2)}) : 1 \le i \le l\}$$
		
		\subsubsection*{Second step}
		
		We now show that if $\refl(\gamma) \maj \refl(\zeta)$ then $\gamma \le \zeta$. Instead of proving this we prove the contrapositive, namely:
		If $\gamma > \zeta$ then $\refl(\gamma) \not\sqsubseteq^{\mathtt{maj}} \refl(\zeta)$.

		Let $\gamma \in \CNF(\omega^{\omega^{d-1}})$ such that $\gamma := \omega^{\beta_1} + \omega^{\beta_2} + \dots + \omega^{\beta_p}$ and $\beta_1 \ge \beta_2 \ge \dots \ge \beta_p$. 
		Further let each $\beta_i: = \omega^{d-2} \cdot c_{(i,d-2)} + \omega^{d-3} \cdot c_{(i,d-3)} + \dots + \omega^{0} \cdot c_{(i,0)}$. Let $\zeta \in \CNF(\omega^{\omega^{d-1}})$ such that $\zeta := \omega^{\eta_1} + \omega^{\eta_2} + \dots + \omega^{\eta_q}$ and $\eta_1 \ge \eta_2 \ge \dots \ge \eta_q$. Further let each 
		$\eta_i := \omega^{d-2} \cdot e_{(i,d-2)} + \omega^{d-3} \cdot e_{(i,d-3)} + \dots + \omega^{0} \cdot e_{(i,0)}$. Suppose $\gamma > \zeta$. Hence, there exists $i \in \{1,\dots,p\}$ such that
		\begin{itemize}
			\item Either $\beta_i > \eta_i$ (or) $i > q$ and
			\item $\forall j$ such that $0 \le j < \min(i,q), \beta_j = \eta_j$
		\end{itemize}
		
		Let $x := (i,c_{(i,0)},c_{(i,1)},\dots,c_{(i,d-2)})$.
		By construction of the map $\refl$ we have that $x \in \refl(\gamma)$. For every $j \in \{1,\dots,q\}$, let $y_j := (j,e_{(j,0)},e_{(j,1)},\dots,e_{(j,d-2)})$. By construction of the map $\refl$ we have that $\refl(\zeta) = \{y_1,\dots,y_q\}$. We will now show that $x \nleq_{\nn^d} y_j$ for each $j$. We consider two cases:
		
		\begin{itemize}
			\item \emph{Case 1: } $j < i$. Therefore $\eta_j = \beta_j$. Hence $y_j := (j,c_{(j,0)},$ $\dots,c_{(j,d-2)})$. Since $j < i$, we have that $x \nleq_{\nn^d} y_j$. 
			\item \emph{Case 2: } $j \ge i$. Therefore $\beta_i > \eta_i \ge \eta_j$. Suppose $x \le_{\nn^d} y_j$. Hence $(i,c_{(i,0)},\dots,c_{(i,d-2)}) \le (j,e_{(j,0)},\dots,e_{(j,d-2)})$ and so
			$(c_{(i,0)},\dots,c_{(i,d-2)}) \le (e_{(j,0)},\dots,e_{(j,d-2)})$. But this means that $\beta_i \le \eta_j$ which leads to a contradiction. Hence we have that $x \nleq_{\nn^d} y_j$.
		\end{itemize}

		Therefore $x \nleq_{\nn^d} y_j$ for every $j$ and so we have $\refl(\gamma) \not\sqsubseteq^{\mathtt{maj}} \refl(\zeta)$. 
		
		\subsubsection*{Third step}
		
		We now show that $|\refl(\gamma)|_{\pow(\nn^d)} \le \varphi(N\gamma)$. \ Let $\gamma \in \CNF(\omega^{\omega^{d-1}})$ such that the Cantor normal form of $\gamma$ is 
		$\omega^{\beta_1} + \omega^{\beta_2} + \dots \omega^{\beta_l}$.
		Further let each $\beta_i := \omega^{d-2} \cdot c_{(i,d-2)} + \omega^{d-3} \cdot c_{(i,d-3)} + \dots + \omega^{0} \cdot c_{(i,0)}$. It is clear that 
		\begin{equation} \label{eqn:reflgamma}
		|\refl(\gamma)|_{\pow(\nn^d)} = \max(l, \{c_{(i,j)}\}^{1 \le i \le l}_ {0 \le j \le d-2})
		\end{equation}
		
		Suppose $\gamma$ in the strict form looks like:
		$\omega^{\gamma_1} \cdot e_1 + \omega^{\gamma_2} \cdot e_2 + \dots + \omega^{\gamma_m} \cdot e_m$ where $\gamma_1 > \gamma_2 > \dots > \gamma_m$ and each $e_i > 0$. Notice that $l = \sum_{i=1}^m e_i$. Further it is also easy to observe that for all $i \in \{1,\dots,m\}$, there exists $j \in \{1,\dots,l\} \text{ such that } \gamma_i = \beta_j$. With this observation, just unraveling the definition of the norm function $N$ implies that
		\begin{equation} \label{eqn:gamma}
		N\gamma = \max(d-2,\{e_i\}^{1 \le i \le m}, \ \{c_{i,j}\}^{1 \le i \le l}_{0 \le j \le d-2})	
		\end{equation}
		
		Since each $\gamma_i \in$ CNF($\omega^{d-1}$), we can write each $\gamma_i$ as
		$\omega^{d-2} \cdot c'_{(i,d-2)} + \omega^{d-3} \cdot c'_{(i,d-3)} + \dots + \omega^0 \cdot c'_{(i,0)}$ where each $c'_{(i,j)} \le N\gamma_i \le N\gamma$. Notice that each $\gamma_i$ is uniquely determined by its coefficients $(c'_{(i,0)},\dots,c'_{(i,d-2)})$, i.e., if $\gamma_i \neq \gamma_j$ then $(c'_{(i,0)},\dots,$ $c'_{(i,d-2)}) \neq (c'_{(j,0)},\dots,c'_{(j,d-2)})$. Therefore we have an injective map from 
		$\{\gamma_i : 1 \le i \le m\}$ to the set $\{x : x \in \nn^{d-1}, \ |x|_{\nn^{d-1}} \le N\gamma \}$. It then follows that $m \le (N\gamma+1)^d$. Hence 
		$$l = \sum_{i=1}^m e_i \le \sum_{i=1}^m N\gamma \le \sum_{i=1}^{(N\gamma+1)^d} N\gamma = \varphi(N\gamma)$$ By equations (\ref{eqn:reflgamma}) and (\ref{eqn:gamma}) this implies that $|\refl(\gamma)|_{\pow(\nn^d)} \le \varphi(N\gamma)$. 
	\end{proof}
	
	Therefore by applying proposition \ref{prop:polyrefl} and theorem \ref{theorem:ordlft} we have,
	\begin{theorem} \label{theorem:lowboundmaj}
		Let $\alpha = \omega^{\omega^{d-1}}$, $\varphi(x) = x(x+1)^d$ and let $n \ge N(\omega^{\omega^{d-1}})$. 
		Then $$g_{\alpha}(n) = L_{\alpha,g}(n)  \le L_{(\pow(\nn^d), \maj), (\varphi \circ g)}(\varphi(n))$$
	\end{theorem}
	
	\section{Upper bound for majoring ordering} \label{section:uppmaj}
	
	In this section we will prove upper bounds on the length of controlled bad sequences for the majoring ordering over $\pow(\nn^d)$. The upper bounds are proven by following the framework established by Schmitz and Schnoebelen in a series of papers(\cite{ICALP},\cite{LICS},\cite{Ordinals}) to prove upper bounds for various well-quasi orders.
	
	We consider the family of nwqos obtained from\\ $\{(\pow(\nn^d),\maj)\}_{d > 0}$ and $\{\Gamma_d\}_{d \in \{0,1\}}$ by taking disjoint sums and cartesian products. We call this family of nwqos the \emph{majoring powerset nwqos}. 
	From now on, we will denote majoring powerset nwqos as a triple $(A,\lemaj_A,|\cdot|_A)$ where $A$ is the domain of the nwqo, $\lemaj_A$ is the underlying order and $|\cdot|_A$ is the norm.
	
	Similar to the proof of upper bounds for the subword ordering in \cite{ICALP}, we introduce an ordinal notation for each majoring powerset nwqo, called the type of that nwqo. The type of a nwqo will turn out to be useful in bounding the corresponding length function using subrecursive hierarchies. 
	
	Notice that if $\alpha \in \CNF(\omega^{\omega^{\omega}})$ then
	$\alpha$ can always be decomposed as $\alpha = \bigoplus_{i=1}^m \bigotimes_{j=1}^{j_i} \omega^{\omega^{d_{i,j}}}$. (Here the empty product is taken to be $1$ and the empty sum is taken to be $0$).
	We now map each majoring powerset nwqo to an ordinal in $\CNF(\omega^{\omega^{\omega}})$ as follows:
	$$o(\Gamma_0) = 0, \qquad o(\Gamma_1) = 1, \qquad o(\pow(\nn^d)) = \omega^{\omega^{d-1}}$$
	$$o(A + B) = o(A) \oplus o(B), \qquad o(A \times B) = o(A) \otimes o(B)$$
	
	Also with each ordinal $\alpha \in$ CNF$(\omega^{\omega^{\omega}})$ we can associate a canonical majoring powerset nwqo, which we will denote by $C(\alpha)$.
	$$C(0) = \Gamma_0, \qquad C(1) = \Gamma_1, \qquad C(\omega^{\omega^d}) = \pow(\nn^{d+1})$$
	$$C(\alpha \oplus \beta) = C(\alpha) + C(\beta), \qquad C(\alpha \otimes \beta) = C(\alpha) \times C(\beta)$$
	
	It can be easily seen that the operators $o$ and $C$ are bijective inverses of each other (modulo isomorphism of nwqos).
	
	\subsection*{Derivatives}
	The next step is to define a \emph{derivative} operator for ordinals. To this end, for each $n \in \nn$, we define a $D_n$ operator as follows:
	$$D_n(k) = k-1, \qquad D_n(\omega) = n+1, \qquad D_n(\omega^{\omega^d}) = \omega^{\omega^{d-1} \cdot (d+1)n}$$
	$$D_n(\omega^{\omega^{p_1} + \omega^{p_2} + \dots + \omega^{p_k}}) = \bigoplus_{i=1}^k \left(D_n(\omega^{\omega^{p_i}}) \otimes \bigotimes_{j \neq i} \omega^{\omega^{p_j}}\right)$$
	
	Using this operator, we define a $\partial_n$ operator as follows:
	$$\partial_n \left(\sum_{i=1}^m \omega^{\beta_i} \right) = \left\{D_n(\omega^{\beta_i}) \oplus  \bigoplus_{j \neq i} \omega^{\beta_j} \; | \; i = 1,\dots,m \right \}$$
	
	Notice that if $\alpha = \omega^{\beta}$ then $\partial_n(\alpha) = \{D_n(\alpha)\}$. 
	\begin{prop} \label{prop:lessthan} (See \ref{proof:lessthan})
		If $\beta \in \partial_n(\alpha)$ then $\beta < \alpha$
	\end{prop}
	
	The following theorem lets us forget the actual underlying nwqo and remember only its type. 
	
	\begin{theorem} \label{theorem:wqotoord} (See \ref{proof:wqotoord})
		Let $A$ be a majoring powerset nwqo and let $\alpha = o(A)$.
		If $X \in A_{\le n}$, then there exists 
		$\alpha' \in \partial_n(\alpha)$ such that there exists a nwqo reflection $r : A/X \hri C(\alpha')$.
	\end{theorem}
	
	Since $o$ and $C$ are inverses of each other, by combining the descent equation and theorem \ref{theorem:wqotoord} we get,
	\begin{lemma} \label{lemma:wqotoordeqn} (See \ref{proof:wqotoordeqn})
		\[L_{C(\alpha),g}(n) \le \max_{\alpha' \in \partial_n (\alpha)} \{1 + L_{C(\alpha'),g}(g(n))\}\]
	\end{lemma}
	
	\subsection*{Upper bound using subrecursive hierarchies}
	
	Given $\alpha \in \CNF(\omega^{\omega^{\omega}})$ define
	$$M_{\alpha,g}(n) = \max_{\alpha' \in \partial_n(\alpha)} \{1 + M_{\alpha',g}(g(n))\}$$
	
	From the definition of $M_{\alpha}(n)$ and lemma \ref{lemma:wqotoordeqn}, it is clear that $L_{C(\alpha),g}(n) \le M_{\alpha,g}(n)$ or in other words,
	$L_{A,g}(n) \le M_{o(A),g}(n)$ for any majoring powerset nwqo $A$.
	Therefore, in what follows, we will concentrate on proving upper bounds for $M_{\alpha,g}(n)$.
	
	Let $\alpha \in \CNF(\omega^{\omega^{\omega}})$. We will say that $\alpha$ is $k$-lean if $N\alpha \le k$. Let $h(x) = 4x \cdot g(x)$ where $g$ is the control function. We have the following important theorem:
	\begin{theorem} \label{theorem:uppmaj} (See \ref{proof:uppmaj})
		If $\alpha$ is $k$-lean and $n > 0$ then $M_{\alpha,g}(n) \le h_{\alpha}(4kn)$
	\end{theorem}
	
	Using theorem \ref{theorem:uppmaj} and the fact that $L_{A,g}(n) \le M_{o(A),g}(n)$, we have the following:
	\begin{theorem} \label{theorem:uppboundmaj}
		Let $A$ be any majoring powerset nwqo such that $o(A)$ is $k$-lean. Then for $n > 0$, we have $L_{A,g}(n) \le M_{o(A),g}(n) \le h_{o(A)}(4kn)$ where $h(x) = 4x \cdot g(x)$.
	\end{theorem}
	
	In particular,
	\begin{corollary} \label{theorem:mainuppboundmaj}
		Let $\alpha = \omega^{\omega^{d-1}}$ and let $n > 0$. Then
		\[L_{(\pow(\nn^{d}),\maj),g}(n) \le h_{\alpha}(4dn)\] where $h(x) = 4x \cdot g(x)$.
	\end{corollary}

	\section{Lower bound for minoring ordering} \label{section:lowmin}
	
	We give a lower bound on the length of controlled bad sequences for the nwqo $(\pow(\nn^d),\smy,|\cdot|_{\pow(\nn^d)})$ by giving a polynomial nwqo reflection from $(\pow(\nn^d),\maj,|\cdot|_{\pow(\nn^d)})$ to 
	$(\pow(\nn^d),\smy,|\cdot|_{\pow(\nn^d)})$.
	
	\begin{lemma} \label{lemma:lowmin}
		There exists a poly. nwqo reflection
		$$\refl : (\pow(\nn^d), \maj, |\cdot|_{\pow(\nn^d)}) \xhri{p} (\pow(\nn^d),\smy, |\cdot|_{\pow(\nn^d)})$$ where $p(x) = d(x+1)$.
	\end{lemma}
	
	\begin{proof}
		Similar to lemma \ref{lemma:lowmaj}, we split the proof into three parts. In the first part, we define the reflection $\refl$. In the second part we show that $\refl(X) \smy \refl(Y) \implies X \maj Y$. Finally, we prove that $|\refl(X)|_{\pow(\nn^d)} \le p(|X|_{\pow(\nn^d)})$.
		
		\subsubsection*{First part}
		The reflection $\refl$ is defined as the following simple map: Given a set $X \in \pow(\nn^d)$, let $\refl(X) := \min(\nn^d \setminus \downarrow X)$.
		
		\subsubsection*{Second part}
		Suppose $\refl(X) \smy \refl(Y)$. By definition this means that
		$\min(\nn^d \setminus \downarrow X) \smy \min(\nn^d \setminus \downarrow Y)$. By the last point of proposition \ref{prop:funfacts} we have that $\nn^d \setminus \downarrow X \smy \nn^d \setminus \downarrow Y$. By the fifth point of proposition \ref{prop:funfacts} it follows that $X \maj Y$.
		
		\subsubsection*{Third part}
		First, we set up some notation.
		Let $\textbf{0}_d$ denote the zero vector in $\nn^d$.
		Given an $x = (x_1,x_2,\dots,x_d)$ $\in \nn^d$ and $i \in \{1,\dots,d\}$, define $x^+_i := (x_1,x_2,\dots,x_{i-1},x_i+1,x_{i+1},\dots,x_d)$. Further, if $x_i > 0$ define $x^-_i := (x_1,x_2,\dots,x_{i-1},$ $x_i-1,x_{i+1},\dots,x_d)$.
		
		We further split this part into two subparts. In the first subpart we prove something about the $\refl$ mapping. In the second part, we use the proposition proven in the first subpart to show that 
		$|\refl(X)|_{\pow(\nn^d)} \le p(|X|_{\pow(\nn^d)})$.
		
		\textbf{First subpart: } Let $X \in \pow(\nn^d)$. We first claim that 
		\begin{equation} \label{eqn:assumption}
		\text{If } y \in \refl(X) \text{ then } y = x^+_i \text{ for some } x \in X \text{ and some } i	
		\end{equation}
		
		Let $y \in \refl(X) = \min(\nn^d \setminus \downarrow X)$. Therefore, $y \nleq_{\nn^d} x$ for any $x \in X$. In particular $y \neq \textbf{0}_d$ and so there exists $i$ such that $y_i \neq 0$. 
		Suppose $y^-_i \in \nn^d \setminus \downarrow X$. Since $y^-_i \le_{\nn^d} y$, it then follows that $y \notin \min(\nn^d \setminus \downarrow X) = \refl(X)$, leading to a contradiction. 
		
		Hence, if $y \in \refl(X)$ then there exists $i$ such that $y_i > 0$ and $y^-_i \notin \nn^d \setminus \downarrow X$. Therefore $\exists x \in X$ such that $y^-_i \le_{\nn^d} x$. Since $y \in \refl(X) = \min(\nn^d \setminus \downarrow X)$ it follows that $y \nleq_{\nn^d} x$. The only way in which we can have $y^-_i \le_{\nn^d} x$ but $y \nleq_{\nn^d} x$ is when $y = x^+_i$, which proves that (\ref{eqn:assumption}) is true.
		
		\textbf{Second subpart: } Let $X^+ := \{x^+_i : x \in X, \ 1 \le i \le d \}$. By (\ref{eqn:assumption}) it is clear that $\refl(X) \subseteq X^+$ and so
		\begin{equation} \label{equation7}
		|\refl(X)|_{\pow(\nn^d)} \le |X^+|_{\pow(\nn^d)}
		\end{equation} 
		
		We proceed to bound $|X^+|_{\pow(\nn^d)}$. To do so, we only need to bound the norm of each element in $X^+$ and the cardinality of $X^+$. 
		By construction, it is easy to see that if $y \in X^+$, then $|y|_{\nn^d} \le |X|_{\pow(\nn^d)} + 1$. Further, by definition of $X^+$, we have $\card(X^+) \le d (\card (X)) \le d(|X|_{\pow(\nn^d)})$.  It then follows that
		\begin{equation} \label{equation8}
		|X^+|_{\pow(\nn^d)} \le d(|X|_{\pow(\nn^d)} + 1)
		\end{equation}
		
		By equations \ref{equation7} and \ref{equation8} it follows that $|\refl(X)|_{\pow(\nn^d)} \le p(|X|_{\pow(\nn^d)})$ which proves the lemma.
	\end{proof}
	
	Let $\varphi(x) = x(x+1)^d$ and let $g_{\varphi} = \varphi \circ g$. Since $\refl$ is a polynomial nwqo reflection, by proposition \ref{prop:polyrefl} and theorem \ref{theorem:lowboundmaj} we have
	
	\begin{theorem} \label{theorem:lowboundmin}
		Let $\alpha = \omega^{\omega^{d-1}}$ and let $n \ge N(\omega^{\omega^{d-1}})$. Then $$g_{\alpha}(n) \le L_{(\pow(\nn^d),\maj),g_{\varphi}}(\varphi(n)) \le L_{(\pow(\nn^d),\smy),(p \circ g_{\varphi})}(p(\varphi(n)))$$
	\end{theorem}

	\section{Upper bound for minoring ordering} \label{section:uppmin}
	
	For the rest of this section, we assume that $d \ge 1$ is fixed.
	Let $(P_i,\lemaj_{P_i},|\cdot|_{P_i})$ be the \emph{majoring powerset nwqo} obtained by taking cartesian product of $(\pow(\nn^i),\maj,|\cdot|_{\pow(\nn^i)})$ with itself $d \choose i$ times, i.e., $P_i = \pow(\nn^i)^{d \choose i}$.
	
	Let $(A_d,\lemaj_{A_d},|\cdot|_{A_d})$ be the \emph{majoring powerset nwqo} formed by taking cartesian product of $P_1,P_2,\dots,P_d$, i.e.,
	$A_d =$ \\ $P_1 \times P_2 \times \cdots \times P_d = $ $\prod_{i=1}^d \pow(\nn^i)^{d \choose i}$. Since $A_d$ is a majoring powerset nwqo, it has an associated \emph{order type} $o(A_d)$ which can be easily seen to be $\bigotimes_{i=1}^d \omega^{(\omega^{i-1}) \cdot {d \choose i}}$. Further it is easy to notice that $o(A_d)$ is $d2^d$-lean.
	
	Having introduced $A_d$, we prove upper bounds on the length of controlled bad sequences for the minoring ordering on $\pow(\nn^d)$ by providing a polynomial nwqo reflection to $A_d$. The reflection that we provide will be a map from $(\pow(\nn^d) \setminus \emptyset,\smy)$ to $A_d$. However, this can be easily converted to an upper bound for $(\pow(\nn^d),\smy)$, thanks to the following proposition:
	
	\begin{prop} \label{prop:empty} 
		\begin{align*}
			L_{(\pow(\nn^d),\smy),g}(n) &= 1 + L_{(\pow(\nn^d)\setminus\emptyset,\smy),g}(g(n)) \\&\le L_{(\pow(\nn^d)\setminus\emptyset,\smy),g}(g(n)+1)
		\end{align*}
	\end{prop}
	
	\begin{proof}
		Notice that for any subset $X \in \pow(\nn^d), X \smy \emptyset$ and so $\pow(\nn^d)/X \subseteq \pow(\nn^d)/\emptyset$. Since $X \smy \emptyset$ for any subset $X$, it follows that $\pow(\nn^d)/\emptyset = \pow(\nn^d)\setminus \emptyset$. 
		Combining these two and applying the descent equation we get,
		\begin{align*}
			L_{(\pow(\nn^d),\smy),g}(n) &= \max_{|X|_{\pow(\nn^d)} \le n}\{1 +  L_{(\pow(\nn^d)/X,\smy),g}(g(n))\} \\&= 1 + L_{(\pow(\nn^d)\setminus\emptyset,\smy),g}(g(n))	
		\end{align*}
		This proves the first equality.
		
		The second inequality is true for the following reason:
		Let $X_0,X_1,\dots,X_l$ be a $(g,g(n))$ controlled bad sequence in $\pow(\nn^d) \setminus \emptyset$. 
		By the last point of proposition \ref{prop:funfacts}, we can assume that $X_i = \min(X_i)$ for each $i$. 
		Let $x := (a_1,a_2,\dots,a_d) \in X_0$.  Construct $x' := (a_1+1,a_2,\dots,a_d)$ and let $X'_0 := (X_0 \setminus \{x\}) \cup \{x'\}$. 
		It can be easily verified that $X'_0,X_0,X_1,\dots,X_l$ is a $(g,g(n)+1)$ controlled bad sequence.
	\end{proof}
	
	Therefore, in what follows, it suffices to focus on $(\pow(\nn^d) \setminus \emptyset,\smy)$. We have the following lemma:
	
	\begin{lemma} \label{lemma:imporlemma} (See \ref{proof:imporlemma})
		There exists a poly. nwqo reflection 
		\[\refl: (\pow(\nn^d) \setminus \emptyset,\smy,|\cdot|_{\pow(\nn^d)}) \xhri{q} (A_d,\lemaj_{A_d},|\cdot|_{A_d})\] where $q(x) = (x+1)^d$.
	\end{lemma}
	
	
	\noindent \textit{Proof sketch.} \
	We present the proof for the case when $d = 2$ and
	then sketch how the proof can be generalised to higher dimensions.
	
	Let us consider $(\pow(\nn^2) \setminus \emptyset,\smy)$ and let $X,Y \in \pow(\nn^2) \setminus \emptyset$.
	By proposition \ref{prop:funfacts} $X \smy Y$ iff $\nn^2 \setminus \uparrow X \maj \nn^2 \setminus \uparrow Y$. Let $\comp(X) := \nn^2 \setminus \uparrow X$ and $\comp(Y) := \nn^2 \setminus \uparrow Y$. Notice that since $X \neq \emptyset$ and $Y \neq \emptyset$, it follows that $\downarrow \comp(X) \neq \nn^2$ and $\downarrow \comp(Y) \neq \nn^2$.
	Therefore there exists $n_X$ and $n_Y$ such that $(n_X,n_X) \notin \downarrow\comp(X)$ and $(n_Y,n_Y) \notin \downarrow\comp(Y)$.
	
	Unfortunately $\comp(X)$ and $\comp(Y)$ might be infinite and so we cannot use the results proved in section \ref{section:uppmaj}. However, we will see that we can ``compress'' the sets $\comp(X)$ and $\comp(Y)$ such that the compressed finite sets preserve the order between $\comp(X)$ and $\comp(Y)$.
	
	Suppose, for some $x \in \nn$, there are infinitely many elements of the form $(x,n_1),(x,n_2),(x,n_3),\dots$ in the set $\comp(X)$. 
	We need not store all these elements, but rather only store that there are infinitely many elements in $\comp(X)$ such that their first co-ordinate is $x$. 
	In accordance with this intuition, we define 
	$S_1^X := \{x : \text{ there exists infinitely many } n \text{ such that }\\ 
	(x,n) \in \comp(X)\}$.
	Notice that $S_1^X$ is a subset of $\nn$.
	Similarly, we define $S_2^X := \{x' : \text{ there exists infinitely many } n \text{ such that }\\ 
	(n,x') \in \comp(X)\}$.
	To complement these two sets, we now define
	$S_3^X := \{(x,x') \in \comp(X) : x \notin S_1^X \text{ and } x' \notin S_2^X\}$.
	We then consider the tuple $(S_1^X,S_2^X,S_3^X)$. Notice that if $(x,x') \in \comp(X)$ then
	either $x \in S_1^X$ or $x' \in S_2^X$ or $(x,x') \in S_3^X$.
	It is then quite easy to see that if $S_1^X \maj_{\pow(\nn)} S_1^Y$ and $S_2^X \maj_{\pow(\nn)} S_2^Y$ and
	$S_3^X \maj_{\pow(\nn^2)} S_3^Y$ then $\comp(X) \maj \comp(Y)$ and so 
	$X \smy Y$.
	
	However it is not clear that each of the sets $S_1^X, S_2^X$ and
	$S_3^X$ are indeed finite. To prove this, first recall that there exists $n_X \in \nn$ such that $(n_X,n_X) \notin \downarrow \comp(X)$.
	
	Suppose $S_1^X$ is infinite. By definition this means that there are infinitely many numbers $x_1,x_2,\dots$ such that for each $x_i$ there are infinitely many elements in $\comp(X)$ with first co-ordinate $x_i$. Pick an $x_i$ such that $x_i \ge n_X$. Now by definition of $S_1^X$ we can pick a $n_i \ge n_X$ such that $(x_i,n_i) \in \comp(X)$. 
	However this means that $(n_X,n_X) \in \downarrow \comp(X)$ which leads to a contradiction. Similar arguments also show that $S_2^X$ is infinite.
	
	Suppose $S_3^X$ is infinite. Since $(n_X,n_X) \notin \downarrow \comp(X)$ it follows that $(n_X,n_X) \notin \downarrow S_3^X$ as well. 
	Hence for every element $(x,y) \in S_3^X$ either $x < n_X$ or $y < n_X$. 
	This indicates that if there are infinitely many elements in $S_3^X$ then there exists $x \in \nn$ such that either there are infinitely many elements in $S_3^X$ with their first co-ordinate as $x$ or there are infinitely many elements with their second co-ordinate as $x$. 
	In either case, by definition of $S_3^X$ we will reach a contradiction.
	
	Finally, we also have to show that $|(S_1^X,S_2^X,S_3^X)|_{A_2} \le (|X|_{\pow(\nn^2)}+1)^2$. First we show that if an element belongs to $S_1^X$ or $S_2^X$ or $S_3^X$ then its norm is bounded by $|X|_{\pow(\nn^2)}$.
	
	Suppose $x \in S_1^X$ and $x > |X|_{\pow(\nn^2)}$. Since $x \in S_1^X$ it follows that there exists $n \ge n_X$ such that $(x,n) \in \comp(X)$.
	Since $(x,n) \in \comp(X) = \nn^2 \setminus \uparrow X$ it follows that 
	for all $(y,m) \in X$ it is the case that $(y,m) \nleq_{\nn^2} (x,n)$. 
	Since $x > |X|_{\pow(\nn^2)} \ge y$ it follows that $m > n$. Hence $(y,m) \nleq_{\nn^2} (x+n+1,n)$ as well. 
	Since this is true for every $(y,m) \in X$ it follows that $(x+n+1,n) \notin \uparrow X$ and so $(x+n+1,n) \in \comp(X)$. 
	Since $(x+n+1,n) \ge_{\nn^2} (n_X,n_X)$ it follows that $(n_X,n_X) \in \downarrow \comp(X)$ which leads to a contradiction. 
	Hence if $x \in S_1^X$ then $x \le |X|_{\pow(\nn^2)}$. 
	A similar argument holds for $S_2^X$ as well.
	
	Suppose $(x,y) \in S_3^X$ and $x > |X|_{\pow(\nn^2)}$. 
	Since $(x,y) \in S_3^X$ there are only finitely many elements in $\comp(X)$ with $y$ as their second co-ordinate. 
	Hence we can find a $n_y$ such that if $n \ge n_y$ then $(n,y) \notin \comp(X)$. 
	Now similar to the case of $S_1^X$ we can now show that $(x+n_y+1,x') \in \comp(X)$ leading to a contradiction. 
	A similar argument is employed when $y > |X|_{\pow(\nn^2)}$. 
	
	Since the norms of the elements of $S_1^X$, $S_2^X$ and $S_3^X$ are bounded by $|X|_{\pow(\nn^2)}$, it follows that their cardinalities are bounded by $(|X|_{\pow(\nn^2)}+1)^2$. Hence the norms of $S_1^X$, $S_2^X$ and $S_3^X$ are each bounded by $(|X|_{\pow(\nn^2)}+1)^2$, which proves our claim.

	We now sketch the construction for the general case of 
	higher dimensions, i.e, when the dimension $d \ge 2$.
	Notice that the set $S_1^X$, as defined for the case of $d = 2$,
	can be stated in the following manner as well: It is the set of all $x$ such that if we fix the first co-ordinate to be $x$ and then project $\comp(X)$ to the second axis, the downward closure of the projection is $\nn$. 
	Hence if we want to prove the lemma for $d=3$, one way to define $S_1^X$ would be: The set of all $x$ such that if we fix the first co-ordinate to be $x$ and then project $\comp(X)$ on the other two axes, the downward closure of the projection is $\nn^2$. 
	In a similar fashion, we can fill in $S_2^X$ and $S_3^X$ by fixing the second co-ordinate and the third co-ordinate. 
	For $S_4^X$ we fix the first and the second co-ordinates and check if the downward closure of the resulting projection is $\nn$ and so on. 
	Then we define the reflection to be $(S_1^X,\dots,S_7^X)$. The reflection for the general case also follows a similar pattern.\\

	Using lemma \ref{lemma:imporlemma}, we can now state upper bounds for the minoring ordering. Let $(\pow(\nn^d)^k,\ \lesmy_{\pow(\nn^d)^k})$ be the nwqo obtained by 
	taking the cartesian product of $(\pow(\nn^d),\smy)$ with itself $k$ times.
	Let $(A_d^k,\lemaj_{A_d^k})$ be the majoring powerset nwqo obtained by taking cartesian product of $(A_d,\lemaj_{A_d})$ with itself $k$ times.
	The following theorem is stated in a way such that it is useful for our applications.
	
	\begin{theorem} \label{theorem:uppboundmin}
		Let $\alpha = \omega^{\omega^{d-1} \cdot (2^d \cdot k)}$ and let $n$ be sufficiently large. There exists a constant $c$ (depending only on $d$ and $k$) such that $$L_{(\pow(\nn^d)^k, \ \lesmy_{\pow(\nn^d)^k}),g}(n) 
		\le t_{\alpha}(c \cdot g(n)^{2d})$$
		where $t(x) = 4kx \cdot q(g(x))$ and $q(x) = (x+1)^d$.
	\end{theorem}
	
	\begin{proof}
		Let $t(x) := 4kx \cdot q(g(x))$. Notice that if $g$ is a strictly increasing inflationary function, then the same is true for $t$. 
		
		Let $\emptyset_k$ denote the tuple $(\overbrace{\emptyset,\dots,\emptyset}^{\text{ k times }})$. 
		The proof of proposition \ref{prop:empty} can be easily modified to prove that
		$$L_{(\pow(\nn^d)^k,\ \lesmy_{\pow(\nn^d)^k}),g}(n) \le L_{(\pow(\nn^d)^k \setminus \emptyset_k, \ \lesmy_{\pow(\nn^d)^k}),g}(g(n)+1)$$
		
		By proposition \ref{prop:precongruence} and lemma \ref{lemma:imporlemma} 
		we have a reflection $(\pow(\nn^d)^k \setminus \emptyset,\ \lesmy_{\pow(\nn^d)^k}) \xhri{k \cdot q} (A_d^k,\lemaj_{A_d^k})$.
		Combining proposition~\ref{prop:polyrefl} and noticing
		that for large enough $n$, we have $q(g(n)+1) \le g(n)^{2d}$,
		we get,
		$$L_{(\pow(\nn^d)^k \setminus \emptyset_k, \ \lesmy_{\pow(\nn^d)^k}),g}(g(n)+1) \le L_{(A_d^k,\lemaj_{A_d^k}),((k \cdot q) \circ g)}(k \cdot g(n)^{2d})	
		$$
		
		Notice that $o(A_d^k) = \bigotimes_{j=1}^k \left(\bigotimes_{i=1}^d \omega^{\omega^{i-1} \cdot {d \choose i}}\right)$ is $dk2^d$-lean. Hence by theorem \ref{theorem:uppboundmaj} we have 
		$$L_{(A_d^k,\lemaj_{A_d^k}),((k \cdot q) \circ g)}(k \cdot g(n)^{2d}) \le t_{o(A_d)}(4dk^22^d g(n)^{2d})$$
		
		Now $o(A_d) < \omega^{\omega^{d-1} \cdot (2^d \cdot k)}$.  
		It is known that, if $\alpha < \alpha'$ then $h_{\alpha}(n) \le h_{\alpha'}(n)$ for sufficiently large $n$ (See Lemma C.9 of \cite{ICALP} and prop \ref{prop:pointwise}). Hence for sufficiently large $n$,
		$$t_{o(A_d)}(4dk^22^d g(n)^{2d}) \le t_{\omega^{\omega^{d-1} \cdot(2^d \cdot k)}}(4dk^22^d g(n)^{2d})$$ 
		
		Hence letting $c := 4dk^22^d$ and $\alpha := \omega^{\omega^{d-1} \cdot (2^d \cdot k)}$ and combining all the equations, we have,
		$$L_{(\pow(\nn^d)^k,\lesmy_{\pow(\nn^d)^k}),g}(n) \le t_{\alpha}(c \cdot g(n)^{2d})$$
			
	\end{proof}

	\section{Complexity classification} \label{section:complexity}
	
	In this section, we will use the results proved in the previous sections to classify length functions for the majoring and minoring ordering based on \emph{fast-growing complexity classes}. Let $S : \nn \to \nn$ denote the successor function. Let $\{S_{\alpha}\}, \{S^{\alpha}\}, \{F_{\alpha}\}$ denote the Hardy, Cichon and fast-growing hierarchies for the successor function respectively. Notice that $S^{\alpha}(x) = S_{\alpha}(x) + x$ for all $x$ and for all $\alpha < \epsilon_0$.
	
	Using these hierarchies, we define \emph{fast growing function classes} 
	$(\fastf_{\alpha})_{\alpha}$ (See \cite{Lob}, \cite{BeyondElem}). 
	$$\fastf_{\alpha} := \bigcup_{c < \omega} \ \text{FD}(F^c_{\alpha}(n))$$
	Here $\text{FD}(F^c_{\alpha}(n))$ denotes the set of all functions that can be computed by a deterministic Turing machine in time $F^c_{\alpha}(n)$ where $F^c_{\alpha}$ denotes the function 
	that results when $F_{\alpha}$ is applied to itself $c$ times. We remark in passing that $\bigcup_{\alpha < \omega} \fastf_{\alpha}$ already constitutes the set of all primitive recursive functions (See section 2.2.4 of \cite{BeyondElem}).

	For the rest of this section, let $g$ be a fixed strictly increasing and inflationary control function such that $g(x) \ge S(x)$.

	\subsection*{Majoring ordering}
	
	Fix a $d > 1$ and let $\varphi(x) = x(x+1)^d$. Our lower bound for the majoring ordering can be readily translated into a complexity lower bound as follows:
	
	\begin{theorem} \label{theorem:complowboundmaj} (See \ref{proof:complowboundmaj})
		For sufficiently large $n$,
		$$F_{\omega^{d-1}}(n) - n \le L_{(\pow(\nn^d), \maj),\varphi \circ g}(\varphi(n))$$ 
		Also $L_{(\pow(\nn^d), \maj), \varphi \circ g} \notin \fastf_{\alpha}$ for any $\alpha < \omega^{d-1}$.
	\end{theorem}
	
	For upper bounds, we state a general result which will be useful for our applications.
	
	\begin{theorem} \label{theorem:compuppboundmaj} (See \ref{proof:compuppboundmaj})
		Let $g$ be a primitive recursive function and let $A = \pow(\nn^d)^k$ for some numbers $d$ and $k$. Then $L_{(A,\lemaj_A),g}$ is eventually bounded by a function in $\fastf_{(\omega^{d-1})\cdot k}$
	\end{theorem}
	
	\subsection*{Minoring ordering}
	
	Let $p(x) = d(x+1)$. The following is a lower bound for the minoring ordering.
	\begin{theorem}\label{theorem:complowboundmin} (See \ref{proof:complowboundmin})
		For sufficiently large $n$, $$F_{\omega^{d-1}}(n) - n \le L_{(\pow(\nn^d),\smy),p \circ \varphi \circ g}(p(\varphi(n)))$$ Also $L_{(\pow(\nn^d),\smy),p \circ \varphi \circ g} \notin \fastf_{\alpha}$ for any
		$\alpha < \omega^{d-1}$.
	\end{theorem}
	
	We also have the following upper bound.
	\begin{theorem} \label{theorem:compuppboundmin} (See \ref{proof:compuppboundmin})
		Let $g$ be primitive recursive and let $A = \pow(\nn^d)^k$ for some numbers $d$ and $k$. Then $L_{(A,\lesmy_A),g}$ is eventually bounded by a function in $\fastf_{\omega^{d-1} \cdot (2^d \cdot k)}$
	\end{theorem}

	\section{Applications} \label{section:applications}
	
	We use the bounds proven in this paper to provide upper bounds for some problems in automata theory. As a first application, we consider the emptiness problem of incrementing tree counter automata (ITCA) over finite labelled trees \cite{ATRA}. We only provide an informal sketch of the model here. (The reader is referred to \cite{ATRA} for the technical details). Incrementing tree counter automata are finite state automata which operate over trees and have access to counters which it can increment, decrement or test for zero. 
	To avoid undecidability, the counters are also allowed to have \emph{incrementation errors}, i.e., the values of the counters can increase erroneously at any time. 
	Based on the theory of well-structured transition systems,
	the paper \cite{ATRA} gives a decision procedure for the emptiness problem for ITCA from a given initial configuration.
	In \cite{Linearise}, the authors argue that if the number of states $q$ and the number of counters $k$ of the given ITCA are fixed
	and the running time is measured as a function of the inital configuration $v_0$ of the ITCA, then
	the running time of this decision procedure could be upper bounded by a function from $\fastf_{(\omega^k) \cdot q}$.
	Since the paper~\cite{Linearise} uses a different notion of controlled bad sequences compared to ours (and a different well-quasi order than the 
	one constructed in this paper),
	we first revisit and adapt their analysis to our setting. Then we apply our results to 
	obtain better bounds for the running time.

	
	Let $q, k$ be fixed natural numbers.
	Recall that $\Gamma_q$ is the well-quasi order 
	where the domain has $q$ elements such that distinct elements are unordered and the norm of every element is 0.
	By taking cartesian product of $\Gamma_q$ with $\nn^k$
	and then taking the majoring ordering of this resulting construction
	we get a well-quasi order which we will
	denote by $(A, \le_A, |\cdot|_A)$ where $A = \pow(\Gamma_q \times \nn^k)$. 
	Notice that the structure $(A,\le_A,|\cdot|_A)$ is isomorphic to
	the majoring powerset nwqo $(\pow(\nn^k)^q, \lemaj_{\pow(\nn^k)^q}, |\cdot|_{\pow(\nn^k)^q})$.
	Indeed suppose $S \in A$.
	For every $a \in \Gamma_q$, let $S_a := \{v : (a,v) \in S\}$.
	It is then clear that the mapping $S \rightarrow (S_a)_{a \in \Gamma_q}$
	is an reflection from $(A,\le_A,|\cdot|_A)$ to
	$(\pow(\nn^k)^q, \lemaj_{\pow(\nn^k)^q}, |\cdot|_{\pow(\nn^k)^q})$.
	We will use this fact later on.
	
	We now analyse the algorithm given in~\cite{ATRA} for testing
	the emptiness of an ITCA. Let $q$ and $k$ be the number of
	states and the number of counters of the ITCA respectively.
	Let $v_0 \in \Gamma_q \times \nn^k$ be the given initial configuration.
	Let $(A,\le_A,|\cdot|_A)$ be the
	nwqo on the domain $A = \pow(\Gamma_q \times \nn^k)$ as described above.
	The algorithm proceeds by constructing
	a sequence of finite sets $K_0,K_1,\dots$ where
	each $K_i \subseteq A$, 
	$K_0$ is the initial configuration $\{v_0\}$
	and $K_{i+1} = K_i \cup Succ(K_i)$ where $Succ$ is
	the successor function between sets of configurations
	as described in \cite{ATRA}.
	The algorithm then finds the first $m$ such that
	$\uparrow K_m = \uparrow K_{m+1}$ and checks
	if there is an accepting configuration in $\uparrow K_m$.
	The complexity of the algorithm is mainly dominated
	by the length of the sequence $K_0,K_2,\dots,K_m$. 
	Since $m$ is the first index such that $\uparrow K_m = \uparrow K_{m+1}$, we can find a minimal element
	$x_i \in \uparrow K_{i+1} \setminus K_i$ for each $i < m$.
	Consider the sequence $x_0,\dots,x_{m-1}$ over $A$.
	Noticing
	that $x_j \ngeq_A x_i$ if $j > i$,
	we can conclude that $x_0,\dots,x_{m-1}$ is a bad sequence over $A$.
	Further by a careful inspection of the $Succ$ relation
	(as described in~\cite{ATRA}) one can easily establish
	that $x_0,\dots,x_{m-1}$ is a $(g,|v_0|_A)$-controlled sequence
	where $g$ is some primitive recursive function depending on $q$ and $k$.
	Now since the nwqo $(A,\le_A,|\cdot|_A)$ has a reflection into
	 $(\pow(\nn^k)^q, \lemaj_{\pow(\nn^k)^q}, |\cdot|_{\pow(\nn^k)^q})$,
	we can apply Theorem~\ref{theorem:compuppboundmaj} and Proposition~\ref{prop:polyrefl}
	to get,
	
	\begin{prop}
		The time complexity of the emptiness problem for an ITCA with $q$ states and $k$ counters is bounded by a function in $\fastf_{(\omega^{k-1})\cdot q}$
	\end{prop}
	
	As noticed in \cite{Linearise}, the authors of \cite{ATRA} also prove the decidability of emptiness for a class of tree automata operating on finite \emph{data trees} called the \emph{alternating top-down tree one register automata} (ATRA), by providing a PSPACE-reduction to the emptiness problem for ITCA. 
	If the original ATRA had $q$ states, then the constructed ITCA has $k(q) = 2^q - 1 + 2^{4q}$ many counters and $f(q) \in O(2^q)$ many states. Hence, we have
	\begin{prop}
		The time complexity of the emptiness problem for an ATRA with $q$ states is bounded by a function in $\fastf_{(\omega^{k(q)-1})\cdot f(q)}$
	\end{prop}

	As a second application, we consider the emptiness problem for another class of finite data tree automata called the \emph{bottom-up alternating one register data tree automata} (BUDA) (See \cite{BUDA} for a complete description of the model). Apart from having a finite number of states $Q$, the transitions of a BUDA are also defined by a specified finite semigroup $S$. In \cite{BUDA}, the authors prove the decidability of the emptiness problem for BUDA using the theory of well-structured transition systems. Let $q$ and $s$ be the number of states and the size of the finite semigroup of the given BUDA respectively. Let $k = 2^{q+s}$ and $l = 2 q^2 s^2 +1$. The authors construct a wsts corresponding to the given BUDA whose set of configurations can be taken to be $(\pow(\nn^k)^{f(k)}$ (where $f(k)$ is some function in $O(2^k)$) with the underlying order being $\lesmy_{\pow(\nn^k)^{f(k)}}$.
	
	A careful analysis of the decision procedure they describe over this wsts reveals that the algorithm constructs a sequence of finite sets $K_0,K_1,\dots,$ where each $K_i \subseteq (\pow(\nn^k)^{f(k)}$, $K_0$ is the initial configuration $v_0$ and
	$K_{i+1} = K_i \cup Succ(K_i)$ where $Succ$ is the successor
	function between sets of configurations as described by the wsts. 
	The algorithm then finds the first $m$ such that $\uparrow K_m = \uparrow K_{m+l}$ and checks if there is an accepting configuration in $\uparrow K_m$. The complexity of the algorithm is mainly dominated by the length of the sequence $K_0,\dots,K_m,\dots,K_{m+l}$. 
	Since $m$ is the first index such that
	$\uparrow K_m = \uparrow K_{m+l}$, we can find a minimal element $x_i \in \uparrow K_{i+l} \setminus \uparrow K_i$ for each $i < m$.
	Let $p$ be the largest number such that $pl \le m+l$.
	Similar to the analysis performed for the ITCA model,
	we can conclude that $x_0,x_{l},x_{2l},\dots,x_{(p-1)l}$ is a $(g,|v_0|_A)$-controlled sequence where $g$ is a primitive recursive function depending on $k$. Applying theorem \ref{theorem:compuppboundmin} we then get,
	
	\begin{prop}
		The time complexity of the emptiness problem for a BUDA with $q$ states and $s$ elements in the semigroup, is bounded by a function in $\fastf_{\omega^{k-1}\cdot (2^k \cdot f(k))}$ where $k = 2^{q+s}$.
	\end{prop}

	\section{Conclusion}
	
	In this paper, we have proved lower and upper bounds for the length of controlled bad sequences for the majoring and minoring ordering over finite sets of $\nn^k$. The results were obtained by giving the bounds in terms of functions from Cichon hierarchy and using known complexity results, were translated into bounds over the fast-growing hierarchy.
	To the best of our knowledge, this is the first upper bound result for length functions over the minoring ordering of $\pow(\nn^d)$. As an application, we used the results to establish upper bounds for the emptiness problems of three types of automata working on trees.
	
	The bounds on the length function for the majoring ordering on $\pow(\nn^k)$ is easily seen to be tight, which solves a problem left open in \cite{Linearise}. However this is not the case with the bounds for minoring ordering and it might be an interesting question in the future to bridge this gap.

	\begin{acks}                            
		The author is extremely grateful to Prof. Philippe Schnoebelen and Prof. Sylvain Schmitz of LSV, ENS-Saclay Paris for their mentorship and useful discussions regarding the paper. This work was done when the author was an intern at LSV, ENS-Saclay Paris. 
		
		This material is based upon work supported by the French National Research Agency (ANR) grant BRAVAS 
		under grant number
		\grantnum{}{ANR-17-CE40-0028} and also by the Indo-French research unit \grantsponsor{}{UMI Relax}{}. 
		The author also acknowledges the support provided by the 
		ERC advanced grant PaVeS
		under grant number \grantnum{}{787367}.
		
	\end{acks}

	\bibliography{References}
	
	\newpage
	\appendix
	
	The following appendices provide the proofs missing in the paper.

\section{Proofs of basic propositions} \label{section:appbasicprop}

\subsection{Proof of proposition \ref{prop:precongruence}} \label{proof:precongruence}

Let $r: A \xhri{q} B$ and $r' : A' \xhri{q'} B'$ be two polynomial reflections. Define a map $s: A+A' \to B+B'$ where $s(x) = r(x)$ if $x \in A$ and $s(x) = r'(x)$ if $x \in A'$. It is immediately clear that $s$ is the desired polynomial reflection. For cartesian products, the map $p: A \times A' \to B \times B'$ where $p(x,y) = (r(x),r'(y))$ is easily seen to be the desired polynomial reflection.

\subsection{Proof of proposition \ref{prop:polyrefl}} \label{proof:polyrefl}

Let $x_0,\dots,x_l$ be a $(g,n)$-controlled bad sequence in 
$A$. Consider the sequence $r(x_0),\dots,r(x_l)$ in $B$.
Since $r$ is a polynomial nwqo reflection, it follows that $r(x_i) \nleq r(x_j)$ for any $i < j$. 
Let $q$ be a strictly increasing inflationary polynomial such that $p(n) \le q(n)$ for all $n$. 
Therefore, we have $|r(x_i)|_B \le p(|x_i|_A) \le q(|x_i|_A) \le q(g^i(n))$. 
It can now be easily seen by induction on $i \ge 0$ that $q(g^i(n)) \le (q \circ g)^i(q(n))$. Hence we have that $r(x_0),\dots,r(x_l)$ is a $(q \circ g,q(n))$-controlled bad sequence in $B$, which proves our claim.

\subsection{Results about sums, products and reflections}

A few properties of sums and products, which can be easily checked are:

\begin{prop} \label{prop:assoc}
	$$A+B \equiv B+A, \qquad A \times B \equiv B \times A,$$
	$$A+(B+C) \equiv (A+B)+C, \qquad A \times (B \times C) \equiv (A \times B) \times C,$$
	$$\Gamma_0 + A \equiv A, \qquad \Gamma_0 \times A \equiv \Gamma_0, \qquad \Gamma_1 \times A \equiv A$$
	$$(A+B) \times C \equiv (A \times C) + (B \times C)$$
\end{prop}

Further, the following facts about nwqo reflections can be easily verified:
\begin{prop} \label{prop:funfactsrefl}
	$$A+B/(1,x) \hri (A/x) + B, \qquad A+B/(2,x) \hri A + (B/x)$$
	$$(A \times B)/(x,y) \hri [(A/x) \times B] + [A \times (B/x)]$$
	$$A \hri A' \text{ and } B \hri B' \text{ implies } A+B \hri A'+B'$$
	$$A \hri A' \text{ and } B \hri B' \text{ implies } A \times B \hri A' \times B'$$
\end{prop}

\subsection{Results about sub-recursive hierarchies}

The following facts are known about sub-recursive hierarchies: (see Lemma 5.1 and C.9 of \cite{ICALP})
\begin{prop} \label{prop:funfactshier}
	Let $h$ be a strictly increasing inflationary function. 
	For all $\alpha \in \CNF(\omega^{\omega^{\omega}})$ and $x \in \nn$ we have:
	\begin{itemize} \setlength\itemsep{1mm}
		\item $h_{\alpha}(x) \le h_{\alpha}(y)$ if $x \le y$
		\item $h_{\alpha}(x) \le h^{\alpha}(x) - x$
		\item $h^{\omega^{\alpha} \cdot r}(x) = f_{h,\alpha}^r(x)$ for all $r < \omega$.
	\end{itemize}
\end{prop}

\section{Proofs for upper bound of majoring ordering} \label{section:appuppmaj}

\subsection{Proof of proposition \ref{prop:lessthan}} \label{proof:lessthan}

The proposition is clearly true when $\alpha = d$ (or) $\alpha = \omega^{\omega^d}$ for some $d \in \nn$. For the general case,
first observe that the following three statements are true:
\begin{align*}
	\alpha < \alpha' &	\implies \alpha \oplus \beta < \alpha' \oplus \beta\\
	\alpha < \alpha' \text{ and } 0 < \beta &\implies \alpha \otimes \beta < \alpha' \otimes \beta\\
	\bigoplus_{i=1}^m \alpha_i < \omega^{\beta} &\iff \alpha_i < \omega^{\beta} \text{ for all } i	
\end{align*}

Using these statements and ordinal induction, the proposition can be proven for the general case as well.

\subsection{Proof of theorem \ref{theorem:wqotoord}} \label{proof:wqotoord}

For the rest of this section, whenever we want to mention that there exists a nwqo reflection from $A$ to $B$, we simply denote it by 
$A \hri B$.

Let $A$ be a majoring powerset nwqo. Since majoring powerset nwqos are built from $\{\pow(\nn^d)\}_{d > 0}$ and $\{\Gamma_d\}_{d \in \{0,1\}}$ by disjoint sums and cartesian products, we can write $A$ as \\
$\sum_{i=1}^m \prod_{j=1}^{j_i} \pow(\nn^{d_{i,j}})$ where the empty sum is taken to be $\Gamma_0$ and the empty product is taken to be $\Gamma_1$. 

Let $X \in A_{\le n}$. We proceed by induction on the structure of $A$.

\begin{itemize}
	\item Suppose $A$ is finite (i.e $j_i = 0$ for every $i$). Therefore $A = \sum_{i=1}^m \Gamma_1$. If $m = 0$, then $A = \Gamma_0$ and the claim holds trivially. Suppose $m > 0$. Therefore $A \equiv \Gamma_m$ for some $m > 0$. 
	By definition
	$$o(A) = m, \; \partial_n(m) = \{m-1\}, \; C(m-1) = \Gamma_{m-1}$$ and indeed $\Gamma_m/X \equiv \Gamma_{m-1}$.
	
	\item Suppose $A = \pow(\nn^d)$ for some $d$. If $d = 1$, then 
	$$o(A) = \omega, \; \partial_n(\omega) = \{n+1\}, \; C(n+1) = \Gamma_{n+1}$$ 
	
	Suppose $Y$ is a finite subset of $\nn$ such that $Y \in A/X$. By assumption $X \in A_{\le n}$ and so if $x \in X$ then $x \le n$. Since $Y \in A/X$ it follows that if $y \in Y$ then $y < n$. Therefore $Y$ is a subset over 
	$\{0,1,\dots,n-1\}$. We consider the following reflection $\refl: A/X \hri \Gamma_{n+1}$: $\refl(\emptyset) = a_0$
	and $\refl(Y) = a_{i+1}$ if the maximum element in $Y$ is $i$. 
	This can be easily seen to be a nwqo reflection and so we are done.
	
	Suppose $d > 1$. Then $$o(A) = \omega^{\omega^{d-1}}, \; \partial_n(\omega^{\omega^{d-1}}) = \{\omega^{\omega^{d-2} \cdot (dn)}\}$$  $$C(\omega^{\omega^{d-2} \cdot (dn)}) = \pow(\nn^{d-1})^{dn}$$ Therefore to prove the claim, it suffices to exhibit a reflection of the form
	$\refl : \pow(\nn^d)/X \hri \pow(\nn^{d-1})^{dn}$.
	
	Let $Y \in \pow(\nn^d)/X$. Since $X \in A_{\le n}$ it follows that if $x \in X$ then $|x| \le n$. Hence 
	\begin{equation} \label{eqn:xn}
	\text{ If } x \in X \text{ then } x \le (\underbrace{n,\dots,n}_{d \text{ times}})		
	\end{equation}
	
	Since $Y \in \pow(\nn^d)/X$, we have that $X \not\sqsubseteq^{\mathtt{maj}} Y$. This combined with equation (\ref{eqn:xn}) implies that if $y \in Y$ then there exists $i$ such that $y_i < n$. Using this fact, we define our reflection as follows: For $1 \le i \le d, \ 0 \le j \le n-1$, define 
	\begin{multline*}
		Y_i^j = \{(y_1,y_2,\dots,y_{i-1},y_{i+1},y_{i+2},\dots,y_d) :\\ (y_1,y_2,\dots,y_d) \in Y, y_i = j\}
	\end{multline*}
 Now consider the map 
 \[
	 Y \to (Y_1^0, \dots, Y_1^{n-1}, Y_2^0, \dots, Y_2^{n-1}, \dots, Y^0_d, \dots, Y_d^{n-1})
 \]
 It is easy to verify that this is indeed a nwqo reflection from $\pow(\nn^d)/X$ to $\pow(\nn^{d-1})^{dn}$.
	
	\item Suppose $A = \prod_{i=1}^k \pow(\nn^{d_i+1})$ (where $d_i$ can be 0). Then,
	$$o(A) = \bigotimes_{i=1}^k \omega^{\omega^{d_i}} = \omega^{\omega^{d_1} \oplus \dots \oplus \omega^{d_k}}$$
	$$\alpha' := \partial_n(A) = \left\{ \bigoplus_{i=1}^k \left(D_n(\omega^{\omega^{d_i}}) \otimes \bigotimes_{j \neq i} \omega^{\omega^{d_j}}\right) \right\}$$
	$$C(\alpha') = C\left[ \bigoplus_{i=1}^k \left(\overbrace{D_n(\omega^{\omega^{d_i}}}^{\beta_i}) \otimes \overbrace{\bigotimes_{j \neq i} \omega^{\omega^{d_j}}}^{\alpha_i'}\right) \right] = \sum_{i=1}^k C(\beta_i) \times C(\alpha_i')$$
	
	Now $C(\alpha_i') = \prod_{j \neq i} \pow(\nn^{d_j+1})$ and
	$C(\beta_i) = C(D_n(\omega^{\omega^{d_i}}))$.
	Let $X = (X_1,\dots,X_k)\in A_{\le n}$. By proposition \ref{prop:funfactsrefl} we have a reflection $A/X \hri \sum_{i=1}^k \pow(\nn^{d_i+1})/X_i \times \prod_{j \neq i} \pow(\nn^{d_j+1})$. 
	By induction hypothesis, there is a reflection of the form $\pow(\nn^{d_i+1})/X_i \hri C(D_n(\omega^{\omega^{d_i}})) = C(\beta_i)$. Combining these two we get that there is a reflection of the form $A/X \hri \sum_{i=1}^k \pow(\nn^{d_i+1})/X_i \ \times \ \prod_{j \neq i} \pow(\nn^{d_j+1}) \hri \sum_{i=1}^k C(\beta_i) \times C(\alpha_i') = C(\alpha')$.
	
	\item Suppose $A$ is of the form $\sum_{i=1}^k \prod_{j=1}^{j_i} \pow(\nn^{d_{i,j}+1})$. Let $A_i = \prod_{j=1}^{j_i} \pow(\nn^{d_{i,j}+1})$ so that $A = \sum_{i=1}^k A_i$. Notice that if $\alpha' \in \partial_n(o(A_i))$ then $\alpha' = D_n(o(A_i))$.
	Now $$o(A) = \bigoplus_{i=1}^k o(A_i)$$
	and 
	$$\partial_n(o(A)) = \left\{D_n(o(A_i)) \oplus \bigoplus_{j \neq i} o(A_j) \ | \ i = 1,\dots,k\right\}$$
	
	Let $X \in A_{\le n}$ such that $X = (i,X')$ for some $X' \in A_i$. By proposition \ref{prop:funfactsrefl} we know that there exists a reflection $A/X \hri (A_i/X') + \sum_{j \neq i} A_j$.
	Let $\alpha' = D_n(o(A_i)) \oplus \bigoplus_{j \neq i} o(A_j)$. By induction hypothesis we have a reflection $A_i/X' \hri C(D_n(o(A_i)))$. Since $\sum_{j \neq i} A_j\equiv$ \\ $C(\bigoplus_{j \neq i} o(A_j))$ it follows that we have a reflection $A/X \hri (A_i/X') + \sum_{j \neq i} A_j \hri 
	C(D_n(o(A_i)) \oplus $\\ $\bigoplus_{j \neq i} o(A_j)) = C(\alpha')$.
\end{itemize}

\subsection{Proof of theorem \ref{lemma:wqotoordeqn}} \label{proof:wqotoordeqn}

Let $A$ be a majoring powerset nwqo and let $\alpha = o(A)$. The descent equation tells us that
$$L_{A,g}(n) = \max_{X \in A_{\le n}} \{1 + L_{A/X,g}(g(n))\}$$
Since $o$ and $C$ are inverse operators of each other, we can rephrase the descent equation as
$$L_{C(\alpha),g}(n) = \max_{X \in A_{\le n}} \{1 + L_{A/X,g}(g(n))\}$$

By theorem \ref{theorem:wqotoord}, if $X \in A_{\le n}$ then there exists $\alpha' \in \partial_n(\alpha)$ such that $A/X \hookrightarrow C(\alpha')$. Hence by proposition \ref{prop:polyrefl} we have,

$$L_{C(\alpha),g}(n) \le \max_{\alpha' \in \partial_n(\alpha)} \{1 + L_{C(\alpha'),g}(g(n)) \}$$

\subsection{Proof of theorem \ref{theorem:uppmaj}} \label{proof:uppmaj}

Similar to the presentation in \cite{ICALP}, we present some intermediate results before proving the main theorem.\\

A main problem with the Cichon hierarchy is that in general $\alpha < \alpha'$ \emph{does not} imply $h_{\alpha}(x) \le h_{\alpha'}(x)$. To demonstrate this, let $h$ be the successor function and let 
$\alpha = n+2$ and $\alpha' = \omega$. Clearly, $h_{n+2}(n) = n+2$ whereas $h_{\omega}(n) = n+1$. This will quickly prove to be a problem in our arguments for proving upper bounds. To handle this, we introduce the notion of \emph{pointwise at-$x$ ordering} \cite{Cichon}. Given $x \in \nn$, we define the relation $\prec_x$ between ordinals as the smallest transitive relation such that for all $\alpha, \lambda$:
\begin{equation*}
\alpha \prec_x \alpha+1, \qquad \lambda_x \prec_x \lambda
\end{equation*}
Here $\lambda_x$ is the $x^{th}$ term in the fundamental sequence for $\lambda$. The inductive definition of $\prec_x$ implies
\begin{equation*}
\alpha \prec_x \alpha' \iff 
\begin{cases}
\alpha' = \beta + 1 \text{ and } \alpha \preceq_x \beta \text{ or }\\
\alpha' = \lambda \text{ and } \alpha \preceq_x \lambda_x
\end{cases}
\end{equation*}

The following is true about the pointwise ordering:

\begin{prop} \label{prop:funfactspointwise} (See section B.3 of \cite{ICALP})
	$$\prec_0 \ \dots \ \subset \ \prec_x \ \subset \ \prec_{x+1} \ \subset \ \dots \ \subset \ \left(\bigcup_{x \in \nn} \ \prec_x \right) \ = \ < $$
	where $<$ is the usual ordering on ordinals.
\end{prop}

We have already introduced the notion of \emph{leanness} of ordinals, but we reintroduce it here for the sake of completeness. Let $\alpha \in \CNF(\epsilon_0)$. We say that $\alpha$ is $k$-lean if
$N\alpha \le k$. Observe that only 0 is 0-lean and if $\alpha$ is $k$-lean and $\alpha'$ is $k'$-lean, then $\alpha \oplus \alpha'$ is 
$k+k'$-lean. We have the following lemma: 

\begin{lemma} \label{lemma:icalp} (Lemma B.1 from \cite{ICALP})
	Let $\alpha$ be $x$-lean. Then $\alpha < \gamma$ iff $\alpha \preceq_{x} P_x(\gamma)$	
\end{lemma}

The following results are know about the pointwise ordering: 
\begin{prop} \label{prop:pointwise} (See Lemma C.9 of \cite{ICALP})
	Let $h$ be a strictly increasing inflationary function. Then,
	\begin{itemize}
		\item $x \le y \implies h_{\alpha}(x) \le h_{\alpha}(y)$
		\item $\alpha \preceq_x \alpha' \implies h_{\alpha}(x) \le h_{\alpha'}(x)$
		\item If $\alpha$ is $x$-lean and $\alpha < \alpha'$ then
		$\forall y \ge x, \ h_{\alpha}(y) \le h_{\alpha'}(y)$
	\end{itemize}
\end{prop}


Using the notions of leanness and pointwise ordering and the above mentioned results, we prove the following:

\begin{lemma} \label{lemma:lean}
	Let $k,n > 0$. Suppose $\alpha \in \CNF(\omega^{\omega^{\omega}})$ is $k$-lean and $\alpha' \in \partial_n(\alpha)$. Then $\alpha'$ is $2k+(k+1)n$-lean.
\end{lemma}

\begin{proof}
	We first prove that if $\alpha = \omega^{\beta}$ for some $\beta \in \CNF(\omega^{\omega})$ and $\alpha' \in \partial_n(\alpha)$ then
	$\alpha'$ is $k + (k+1)n$-lean. Notice that if $\alpha = \omega^{\beta}$ then $\partial_n(\alpha) = \{D_n(\alpha)\}$.
	
	The claim is clearly true when $\alpha = k$ for some $k \in \nn$ or $\alpha = \omega^{\omega^d}$ for some $d \ge 0$. Suppose 
	$\alpha = \omega^{\omega^{p_1} + \dots + \omega^{p_l}}$.
	Let $\beta := \omega^{p_1} + \dots + \omega^{p_l}$. Further let 
	$\beta$ in strict form be $\sum_{i=1}^m \omega^{q_i} \cdot c_i$, where $q_1 > q_2 > \dots > q_m$.
	
	We distinguish between two cases: Suppose $q_m \neq 0$. Therefore $q_i \neq 0$ for any $i$. In this case, notice that if $\alpha' \in \partial_n(\alpha)$ then
	\begin{align*}
	\alpha' &= \bigoplus_{i=1}^m \left(D_n(\omega^{\omega^{q_i}}) \cdot  c_i \otimes \omega^{\omega^{q_i} \cdot (c_i-1)} \otimes \bigotimes_{j \neq i} \omega^{\omega^{q_j} \cdot c_j}\right)\\
	&= \bigoplus_{i=1}^m \left(\omega^{\beta_i} \cdot c_i\right)
	\end{align*}
	
	where $$\beta_i = \omega^{q_i-1} \cdot (q_i+1)n \oplus \omega^{q_i} \cdot (c_i - 1) \oplus \bigoplus_{j \neq i} \omega^{q_j} \cdot c_j$$
	
	Notice that the coefficients of $\beta_i$ when written in strict form can be only one of the four possible choices: $(q_i+1)n, c_i-1, c_j$ and can also be $c_{i+1}+(q_i+1)n$ if $q_i - q_{i+1} = 1$. In either case notice that $\beta_i$ is $k+(k+1)n$-lean. Further since 
	$q_1 > q_2 > \dots > q_m$ it follows that $\beta_m > \beta_{m-1} > \dots > \beta_1$ and so $\bigoplus_{i=1}^m \left(\omega^{\beta_i} \cdot c_i\right)$ is 
	\emph{the strict form} of $\alpha'$. Since each $c_i \le k$ and each $\beta_i$ is $k+(k+1)n$-lean, it follows that $\alpha'$ is $k+(k+1)n$-lean as well.
	
	Suppose $q_m = 0$. In this case, if $\alpha' \in \partial_n(\alpha)$ then
	$$\alpha' = \left(\bigoplus_{i=1}^{m-1} \omega^{\beta_i} \cdot c_i\right) \oplus \left(\omega^{\beta_m} \cdot (n+1)c_i \right)$$
	
	where 
	\begin{align*}
		\beta_i &= \omega^{q_i-1} \cdot (q_i+1)n \oplus \omega^{q_i} \cdot (c_i - 1) \oplus \bigoplus_{j \neq i} \omega^{q_j} \cdot c_j \text{ (for } i \neq m)\\
		\beta_m &= \omega^{q_i} \cdot (c_i - 1) \oplus \bigoplus_{j \neq i} \omega^{q_j} \cdot c_j	
	\end{align*}
	
	Clearly $\beta_m$ is $k+(k+1)n$-lean. By the same argument given for the previous case, we can conclude that for $i \neq m, \beta_i$ is $k+(k+1)n$-lean. Since 
	$\beta_m > \beta_{m-1} > \dots > \beta_1$ it follows that $\left(\bigoplus_{i=1}^{m-1} \omega^{\beta_i} \cdot c_i\right) \oplus \left(\omega^{\beta_m} \cdot (n+1)c_i \right)$ is the strict form of $\alpha'$. Since $(n+1)c_i \le (n+1)k \le k+(k+1)n$ and since each $\beta_i$ is $k+(k+1)n$-lean, we conclude that $\alpha'$ is $k+(k+1)n$-lean.\\
	
	Now, we come to the general case. Suppose $\alpha = \sum_{i=1}^k \omega^{\beta_i}$. Notice that if $\alpha' \in \partial_n(\alpha)$ then 
	$\alpha' = D_n(\omega^{\beta_i}) \oplus \bigoplus_{j \neq i} \omega^{\beta_j}$ for some $i$. We just proved that $D_n(\omega^{\beta_i})$ is $k+(k+1)n$-lean. Since $\bigoplus_{j \neq i} \omega^{\beta_j}$ is $k$-lean, it follows that $\alpha'$ is $2k+(k+1)n$-lean.
	
\end{proof}

Recall that we took the control function $g$ to be a strictly increasing inflationary function.

\subsubsection*{Proof of theorem \ref{theorem:uppmaj}}

Let $h(x) := 4x \cdot g(x)$ and let $n > 0$. Notice that $h$ is a strictly increasing inflationary function as well.

We prove the theorem by induction on $\alpha$. The claim is clear for $\alpha = 0$. Suppose $\alpha > 0$. Let $\alpha$ be $k$-lean. We have that $k > 0$ and $M_{\alpha,g}(n) = 1 + M_{\alpha',g}(g(n))$ for some $\alpha' \in \partial_n (\alpha)$. 
By Lemma \ref{lemma:lean} we have that $\alpha'$ is $2k + (k+1)n$-lean. Since $n > 0$ we have that $2k + (k+1)n \le 4kn$ and hence $\alpha'$ is $4kn$-lean as well. By proposition \ref{prop:lessthan}, $\alpha' < \alpha$ and so we can apply the induction hypothesis on $\alpha'$. Hence, by induction hypothesis we have,
\begin{equation*}
	M_{\alpha,g}(n) = 1 + M_{\alpha',g}(g(n)) 
	\le 1 + h_{\alpha'}(4 (4kn) \cdot g(n))	
\end{equation*}

Since $g$ is strictly increasing (by assumption) and $h_{\alpha'}$ is strictly increasing (by proposition \ref{prop:funfactshier}) we have
\begin{equation*}
	1 + h_{\alpha'}(4 (4kn) \cdot g(n))	 \le 1 + h_{\alpha'}(4 (4kn) \cdot g(4kn))
\end{equation*}

By definition, $h(x) := 4x \cdot g(x)$ and so
\begin{equation*}
	1 + h_{\alpha'}(4 (4kn) \cdot g(4kn)) = 1 + h_{\alpha'}(h(4kn))
\end{equation*}

Since $\alpha'$ is $4kn$-lean and $\alpha' < \alpha$, by lemma
\ref{lemma:icalp} we have that $\alpha' \preceq_{4kn} P_{4kn}(\alpha)$. Hence by proposition \ref{prop:funfactspointwise}, $\alpha' \preceq_{h(4kn)} P_{4kn}(\alpha)$. Hence by point two of proposition 
\ref{prop:pointwise} we conclude that
\begin{equation*}
	1 + h_{\alpha'}(h(4kn)) \le 1 + h_{P_{4kn}(\alpha)}(h(4kn))
\end{equation*}

By definition of the Cichon hierarchy we have,
\begin{equation*}
	1 + h_{P_{4kn}(\alpha)}(h(4kn)) = h_{\alpha}(4kn)
\end{equation*}

Combining all the equations we now get that
\begin{equation*}
	M_{\alpha,g}(n) \le h_{\alpha}(4kn)
\end{equation*}

thus proving our theorem.

\section{Proofs for upper bound of minoring ordering} \label{section:appuppmin}

\subsection{Proof of lemma \ref{lemma:imporlemma}} \label{proof:imporlemma}

We need new notations for this lemma, which we introduce here.
For every $n \in \nn$, let $\textbf{n}_j$ denote the vector $(n,n,\dots,n) \in \nn^j$. Further, given $z = (z_1,z_2,\dots,z_d) \in \nn^d$ and numbers $1 \le i_1 < i_2 < \dots < i_k \le d$, let 
$$z_{i_1,i_2,\dots,i_k} = (z_{i_1},z_{i_2},\dots,z_{i_k})$$
\begin{multline*}
	z_{-i_1,-i_2,\dots,-i_k} = (z_1,\dots,z_{i_1-1},z_{i_1+1},\dots,
	z_{i_2-1},z_{i_2+1},\\\dots,z_{i_k-1},z_{i_k+1},\dots,z_d)		
\end{multline*}

By convention we let $z_{-1,\dots,-d} = ()$ where $()$ is the empty vector and we let $\nn^0$ denote the singleton set containing the empty vector.

Let $1 \le i_1 < i_2 < \dots < i_k \le d$ and $j_1,j_2,\dots,j_k$ be natural numbers. Given an arbitrary subset $Z$ of $\nn^d$, define
\begin{multline*}
	\proj^Z_{i_1,i_2,\dots,i_k}(j_1,j_2,\dots,j_k) := \{z_{-i_1,-i_2,\dots,-i_k} : 
	z \in Z, \\z_{i_1,\dots,i_k} = (j_1,\dots,j_k) \}	
\end{multline*}

Intuitively $\proj^Z_{i_1,i_2,\dots,i_k}(j_1,j_2,\dots,j_k)$ collects all elements from $Z$ whose $i_1^{th}$ value is $j_1$, $i_2^{th}$ value is $j_2$ and so on and then projects all these elements on the remaining co-ordinates.\\

We now set out to define the desired polynomial reflection. Let $X \in \pow(\nn^d) \setminus \emptyset$ and let $\comp(X) := \nn^d \setminus \uparrow X$. Notice that $\comp(X)$ is downward closed, (i.e),
$\downarrow \comp(X) = \comp(X)$. Also since $X \neq \emptyset$ we have the following very important fact, which we will extensively use:
\begin{equation} \label{eqn:extensive}
\downarrow \comp(X) \neq \nn^d
\end{equation}

For ease of understanding, we break the construction of the polynomial reflection into several steps: As a first step, for every sequence $i_1,i_2,\dots,i_k$ such that $1 \le i_1 < i_2 < \dots < i_k \le d$, we define an element $X_{i_1,\dots,i_k}$ of $\pow(\nn^k)$. In the second step, we show that each of these sets are finite. In the third step, for every $k$ such that $1 \le k \le d$, we define an element $X^k$ of $\pow(\nn^k)^{d \choose k}$ by using the sets defined in the first step.  In the fourth step, we define our reflection to be $\refl(X) = (X_1,\dots,X_d)$ and show that if $\refl(X) \lemaj_{A_d} \refl(Y)$ then $X \smy Y$. In the last step, we will show that $|\refl(X)|_{A_d} \le q(|X|_{\pow(\nn^d)})$ where $q(x) = (x+1)^d$.

\subsubsection*{First step}

Let $1 \le i \le d$. Define
\begin{equation*}
X_i := \{x_i : x \in \comp(X), \ \downarrow \proj^{\comp(X)}_{i}(x_i) = \nn^{d-1}\}
\end{equation*}

Note that following our notation, $x_i$ is the $i^{th}$ co-ordinate of $x$. Intuitively, $x_i \in X_i$ implies that if we pick all elements in $\comp(X)$ such that their $i^{th}$ co-ordinate is $x_i$ and project all these elements on all the other axes, the projection is `scattered everywhere' over $\nn^{d-1}$. Since with respect to the other co-ordinates it `looks like' $\nn^{d-1}$ we can discard all the other co-ordinates of $x$ and remember only the $i^{th}$ co-ordinate.

For the general case, let $i_1,i_2,\dots,i_k$ be a sequence such that $1 \le i_1 < i_2 < \dots < i_k \le d$. Define
\begin{align*}
X_{i_1,i_2,\dots,i_k} := &\{x_{i_1,i_2,\dots,i_k} : x \in \comp(X),\\ &\downarrow \proj^{\comp(X)}_{i_1,\dots,i_k}(x_{i_1,\dots,i_k}) = \nn^{d-k}, \text{ and for any } \\ &\text{ strict subsequence }
j_1,\dots,j_l \text{ of } i_1,\dots,i_k,  \\ &(x_{j_1,\dots,j_l}) \notin X_{j_1,\dots,j_l}\}
\end{align*}

The intuition behind this definition is the same as the previous one, except now we do not allow an element in the set if `some part of it' is already present in a `lower set'.

\subsubsection*{Second step}

We now prove that for every sequence $i_1,i_2,\dots,i_k$ such that $1 \le i_1 < i_2 < \dots < i_k \le d$, the set $X_{i_1,\dots,i_k}$ is a finite set. 
Before proving that we have the following lemmas:

\begin{lemma} \label{lemma:imporlemma1}
	Let $i_1,\dots,i_k$ be such that $1 \le i_1 < \dots < i_k \le d$. Then,
	\begin{equation*}
	\exists n \in \nn, \ \forall (m_1,\dots,m_k) \in X_{i_1,\dots,i_k}, \ \exists j \text{ such that } m_{j} < n
	\end{equation*}
\end{lemma}

\begin{proof}
	For ease of notation, assume (without loss of generality) that $i_1 = 1, i_2 = 2,\dots, i_k = k$. By equation (\ref{eqn:extensive}), $\downarrow \comp(X) \neq \nn^d$ and so there exists $n \in \nn$ such that 
	\begin{equation} \label{eqn:conseq1}
	\textbf{n}_d \notin \downarrow \comp(X)	
	\end{equation}
	
	We now claim that 
	\begin{equation} \label{eqn:claim1}
	\forall (m_1,\dots,m_k) \in X_{1,\dots,k}, \ \exists j \text{ such that } m_j < n 	
	\end{equation}
	
	Suppose (\ref{eqn:claim1}) is false. We have the following series of implications, all of which follow immediately from the definitions introduced so far:
	\begin{align*}
	(\ref{eqn:claim1}) \text{ is false } &\implies \exists (m_1,\dots,m_k) \in X_{1,\dots,k}  \text{ s.t. } (m_1,\dots,m_k) \ge \textbf{n}_k\\
	&\implies \downarrow \proj^{\comp(X)}_{1,\dots,k}(m_1,\dots,m_k) = \nn^{d-k}\\
	&\implies \exists (m_{k+1},\dots,m_d) \in \proj^{\comp(X)}_{1,\dots,k}(m_1,\dots,m_k) \\& \hspace{1cm}\text{ such that } (m_{k+1},\dots,m_d) \ge \textbf{n}_{d-k}\\
	&\implies (m_1,\dots,m_k,m_{k+1},\dots,m_d) \in \comp(X)
	\end{align*}

	Since $(m_1,\dots,m_k) \ge \textbf{n}_k$ and $(m_{k+1},\dots,m_d) \ge \textbf{n}_{d-k}$ it follows that $(m_1,\dots,m_d) \ge \textbf{n}_d$. Since $(m_1,\dots,m_d) \in \comp(X)$ it follows that $\textbf{n}_d \in \downarrow \comp(X)$ which contradicts (\ref{eqn:conseq1}). Therefore (\ref{eqn:claim1}) is true, which proves the lemma.
\end{proof}

Notice that lemma \ref{lemma:imporlemma1} immediately implies that $X_i$ is finite for every $i \in \{1,\dots,d\}$. However to prove the same for the general case requires another lemma which we prove next.

\begin{lemma} \label{lemma:imporlemma2}
	Let $1 \le i_1 < \dots < i_k \le d$ and $1 \le j_1 < \dots < j_l \le k$ such that $1 \le l < k$. Further let $m_1,\dots,m_l$ be natural numbers. Then
	\begin{multline*}
	\exists n \in \nn, \ \forall (m_{l+1},\dots,m_k) \in \proj^{X_{i_1,\dots,i_k}}_{j_1,\dots,j_l}(m_1,\dots,m_l), \\ \exists q \in \{l+1,\dots,k\} \text{ s.t. } m_q < n
	\end{multline*}
\end{lemma}

\begin{proof}
	The proof of this lemma is very similar to the proof of the previous lemma, but for completeness sake we present the proof here. 
	
	To make notation easier in the proof, (without loss of generality) we assume that $i_1 = 1, i_2 = 2, \dots, i_k = k$ and $j_1 = 1, j_2 = 2, \dots, j_l = l$. If $\proj^{X_{1,\dots,k}}_{1,\dots,l}(m_1,\dots,m_l)$ is empty then there is nothing to prove. Hence in the sequel we assume that $\proj^{X_{1,\dots,k}}_{1,\dots,l}(m_1,\dots,m_l)$ is non-empty.
	Let $(x_{l+1},\dots,x_k) \in \proj^{X_{1,\dots,k}}_{1,\dots,l}(m_1,\dots,m_l)$. By definition we have,
	\begin{align*}
	(x_{l+1},\dots,x_k) \in \proj^{X_{1,\dots,k}}_{1,\dots,l}(m_1,\dots,m_l)\\ \implies (m_1,\dots,m_l,x_{l+1},\dots,x_k) &\in X_{1,\dots,k}\\
	\implies \downarrow \proj^{\comp(X)}_{1,\dots,l}(m_1,\dots,m_l) &\neq \nn^{d-l}\\
	\implies \exists n \in \nn \text{ such that }  \textbf{n}_{d-l} \notin \downarrow &\proj^{\comp(X)}_{1,\dots,l}(m_1,\dots,m_l)	
	\end{align*}

	Hence there exists $n \in \nn$ such that 
	\begin{equation} \label{eqn:conseq2}
	\textbf{n}_{d-l} \notin \downarrow \proj^{\comp(X)}_{1,\dots,l}(m_1,\dots,m_l)
	\end{equation}
	
	We now claim that 
	\begin{multline} \label{eqn:claim2}
	\forall (m_{l+1},\dots,m_k) \in \proj^{X_{1,\dots,k}}_{1,\dots,l}(m_1,\dots,m_l), \\ \exists q \in \{l+1,\dots,k\} \text{ such that } m_q < n
	\end{multline}
	
	Suppose (\ref{eqn:claim2}) is false. We make the following series of implications, all of which follow immediately from definitions introduced so far:
	\begin{align*}
	(\ref{eqn:claim2}) \text{ is false } &\implies \exists (m_{l+1},\dots,m_k) \in \proj^{X_{1,\dots,k}}_{1,\dots,l}(m_1,\dots,m_l) \\& \hspace{1cm} \text{ such that } (m_{l+1},\dots,m_k) \ge \textbf{n}_{k-l}\\
	&\implies (m_1,\dots,m_l,m_{l+1},\dots,m_k) \in X_{1,\dots,k}\\
	&\implies \downarrow \proj^{\comp(X)}_{1,\dots,k}(m_1,\dots,m_k) = \nn^{d-k}\\
	&\implies \exists (m_{k+1},\dots,m_d) \in \proj^{\comp(X)}_{1,\dots,k}(m_1,\dots,m_k) \\ & \hspace{1cm} \text{ such that } (m_{k+1},\dots,m_d) \ge \textbf{n}_{d-k}\\
	&\implies (m_1,\dots,m_k,m_{k+1},\dots,m_d) \in \comp(X)\\
	&\implies (m_{l+1},\dots,m_d) \in \proj^{\comp(X)}_{1,\dots,l}(m_1,\dots,m_l)	
	\end{align*}

	Since $(m_{l+1},\dots,m_k) \ge \textbf{n}_{k-l}$ and $(m_{k+1},\dots,m_d) \ge \textbf{n}_{d-k}$ it follows that $(m_{l+1},\dots,m_d) \ge \textbf{n}_{d-l}$. Since $(m_{l+1},\dots,m_d) \in \proj^{\comp(X)}_{1,\dots,l}(m_1,\dots,m_l)$ it follows that $\textbf{n}_{d-k} \in$ \\
	$\downarrow \proj^{\comp(X)}_{1,\dots,l}(m_1,\dots,m_l)$ which contradicts (\ref{eqn:conseq2}). Therefore (\ref{eqn:claim2}) is true, which proves the lemma.
\end{proof}

A consequence of \ref{lemma:imporlemma2} is the following:

\begin{lemma} \label{lemma:imporlemma3}
	Let $1 \le i_1 < \dots < i_k \le d$ and $1 \le j_1 < \dots < j_l \le k$ such that $1 \le l < k$. Further let $m_1,\dots,m_l$ be natural numbers. Then $\proj^{X_{i_1,\dots,i_k}}_{j_1,\dots,j_l}(m_1,\dots,m_l)$ is finite.
\end{lemma}

\begin{proof}
	Once again, to make notation easier in the proof, we assume that $i_1 = 1, i_2 = 2, \dots, i_k = k$ and $j_1 = 1, j_2 = 2, \dots, j_l = l$. 
	
	We fix a $k$ and prove the lemma by backward induction on $l$. Suppose $l = k-1$. Then $\proj^{X_{1,\dots,k}}_{1,\dots,l}(m_1,\dots,m_l)$ is a subset of $\nn$. The statement of lemma \ref{lemma:imporlemma2} immediately implies that $\proj^{X_{1,\dots,k}}_{1,\dots,l}(m_1,\dots,m_l)$ is finite.
	
	Suppose $l < k-1$ and suppose $\proj^{X_{1,\dots,k}}_{1,\dots,l}(m_1,\dots,m_l)$ is infinite. Applying lemma \ref{lemma:imporlemma2} we have,
	\begin{multline} \label{eqn:claim3}
	\exists n \in \nn, \ \forall (x_1,\dots,x_{k-l}) \in \proj^{X_{1,\dots,k}}_{1,\dots,l}(m_1,\dots,m_l), \\ \exists q \in \{1,\dots,k-l\} \text{ s.t. } x_q < n
	\end{multline}
	
	We now partition $\proj^{X_{1,\dots,k}}_{1,\dots,l}(m_1,\dots,m_l)$ as follows: For each $a \in \{1,\dots,k-l\}$ and $b \in \{0,\dots,n-1\}$ we set 
	\begin{equation*}
	P^a_b := \{x : x \in \proj^{X_{1,\dots,k}}_{1,\dots,l}(m_1,\dots,m_l), \ x_a = b\}
	\end{equation*}
	
	Equation (\ref{eqn:claim3}) guarantees us that  \[\proj^{X_{1,\dots,k}}_{1,\dots,l}(m_1,\dots,m_l) = \bigcup_{1 \le a \le k-l, \ 0 \le b \le n-1} P^a_b\] Since we assumed that $\proj^{X_{1,\dots,k}}_{1,\dots,l}(m_1,\dots,m_l)$ is infinite, it follows that at least one $P^a_b$ is infinite. Without loss of generality assume that $a = 1$. 
	
	Notice that if $x,y \in P^1_b$ then $x_1 = y_1 = b$. Since $P^1_b$ is infinite,
	it is then easy to see that $\proj^{P^1_b}_{1}(b)$ is also infinite. It is also easy to see that $\proj^{P^1_b}_{1}(b)$ is the same as $\proj^{X_{1,\dots,k}}_{1,\dots,l,l+1}(m_1,\dots,m_l,b)$. By induction hypothesis,\\ $\proj^{X_{1,\dots,k}}_{1,\dots,l,l+1}(m_1,\dots,m_l,b)$ is finite which leads to a contradiction.
\end{proof}

Using lemma \ref{lemma:imporlemma3} we now prove that 
each set $X_{i_1,\dots,i_k}$ is finite.

\begin{theorem} \label{theorem:importheorem}
	Let $i_1,\dots,i_k$ be such that $1 \le i_1 < \dots < i_k \le d$. Then $X_{i_1,\dots,i_k}$ is a finite set.
\end{theorem}

\begin{proof}
	Once again, for ease of notation we assume that $i_1 = 1,\dots, i_k = k$. 
	
	Suppose $k = 1$, i.e., we only have one number $i_1 = 1$. By definition $X_1$ is a subset of $\nn$. Lemma \ref{lemma:imporlemma1} then immediately implies that $X_1$ is a finite set.
	
	Suppose $k > 1$ and suppose $X_{1,\dots,k}$ is infinite. Applying lemma \ref{lemma:imporlemma1} we have 
	\begin{equation} \label{eqn:claim4}
	\exists n \in \nn, \ \forall (m_1,\dots,m_k) \in X_{1,\dots,k}, \ \exists j \text{ such that } m_j < n
	\end{equation}
	We now partition $X_{1,\dots,k}$ as follows: For each $a \in \{1,\dots,k\}$ and $b \in \{0,\dots,n-1\}$ we set
	\begin{equation*}
	P^a_b := \{(m_1,\dots,m_k): (m_1,\dots,m_k) \in X_{1,\dots,k}, \ m_a = b\}
	\end{equation*}
	
	Equation (\ref{eqn:claim4}) guarantees us that \[X_{1,\dots,k} = \bigcup_{1 \le a \le k, 0 \le b \le n-1} P^a_b\] Since we assumed that $X_{1,\dots,k}$ is infinite it follows that at least one $P^a_b$ is infinite. Without loss of generality assume that $a = 1$.
	
	Notice that if $x,y \in P^1_b$ then $x_1 = y_1 = b$. Since $P^1_b$ is infinite, it is then easy to see that $\proj^{P^1_b}_1(b)$ is infinite. It is also easy to see that $\proj^{P^1_b}_1(b)$ is the same as $\proj^{X_{1,\dots,k}}_{1}(b)$. By lemma \ref{lemma:imporlemma3}, $\proj^{X_{1,\dots,k}}_{1}(b)$ is finite which leads to a contradiction.
\end{proof}

\subsubsection*{Third step}

We now combine the sets defined in the first step to get sets $X_1,\dots,X_d$. To do this, first notice that if we want to define a tuple of size ${d \choose k}$, we can choose to index the positions by 
sequences of the form $i_1,\dots,i_k$ where $1 \le i_1 < i_2 < \dots < i_k \le d$. We adopt this indexing method for the rest of the section.

We now define $X_k$ as an element in $\pow(\nn^k)^{d \choose k}$ where the element in the position indexed by $i_1,\dots,i_k$ is the set
$X_{i_1,\dots,i_k}$.

By theorem \ref{theorem:importheorem}, each $X_{i_1,i_2,\dots,i_k}$ is a finite subset of $\nn^k$ and so each $X^k$ is an element of $\pow(\nn^k)^{d \choose k}$. 

\subsubsection*{Fourth step}

Define $\refl(X) = (X_1,\dots,X_d)$. It is clear that $\refl(X)$ is a function from $\pow(\nn^d)$ to 
$A_d$. We now show that $\refl(X) \lemaj_{A_d} \refl(Y)$ then $X \smy Y$.

\begin{lemma}
	\begin{equation*}
	\refl(X) \lemaj_{A_d} \refl(Y) \implies X \smy Y
	\end{equation*}
\end{lemma}

\begin{proof}
	Let $X,Y$ be such that $\refl(X) \lemaj_{A_d} \refl(Y)$. To prove that $X \smy Y$, by proposition \ref{prop:funfacts} it is enough to prove that 
	$\comp(X) \maj \comp(Y)$. This is what we will prove in the sequel.
	
	To prove that $\comp(X) \maj \comp(Y)$, we have to show that if $x \in \comp(X)$ then there exists $y \in \comp(Y)$ such that $x \le_{\nn^d} y$. In the sequel we fix a $x \in \comp(X)$.
	
	Let $k$ be the least integer such that there exists indices $i_1,\dots,i_k$ such that $\downarrow \proj^{\comp(X)}_{i_1,\dots,i_k}(x_{i_1,\dots,i_k}) = \nn^{d-k}$ and
	for every strict subsequence $j_1,\dots,j_l$ of $i_1,\dots,i_k$ we have \\$\downarrow \proj^{\comp(X)}_{j_1,\dots,j_l}(x_{j_1,\dots,j_l}) \neq \nn^{d-l}$. (Notice that such a $k$ always exists since $\downarrow \proj^{\comp(X)}_{1,\dots,d}(x_{1,\dots,d}) = \nn^0$). Without loss of generality assume that $i_1 = 1, i_2 = 2,\dots, i_k = k$. By definition, it is easily seen that $(x_1,\dots,x_k) \in X_{1,\dots,k}$.
	
	Let $n \in \nn$ be such that $\textbf{n}_{d-k} \ge (x_{k+1},\dots,x_d)$. We now have the following series of implications, each of which can be easily seen to be true:
	\begin{align*}
	(x_1,\dots,x_k) \in X_{1,\dots,k} \text{ and } \refl(X) \lemaj_{A_d} \refl(Y) \\ \implies \exists (y_1,\dots,y_k) \in Y_{1,\dots,k}  \text{ such that } \\ (x_1,\dots,x_k) \le (y_1,\dots,y_k)\\
	\implies \downarrow \proj^{\comp(Y)}_{1,\dots,k}(y_1,\dots,y_k) = \nn^{d-k}&\\
	\implies \exists (y_{k+1},\dots,y_d) \in \proj^{\comp(Y)}_{1,\dots,k}(y_1,&\dots,y_k) \\\text{ such that } (y_{k+1},\dots,y_d) \ge \textbf{n}_{d-k}\\
	\implies (y_1,\dots,y_k,y_{k+1},\dots,y_d) \in \comp(&Y)
	\end{align*}
	
	Since $(y_1,\dots,y_k) \ge (x_1,\dots,x_k)$ and $(y_{k+1},\dots,y_d) \ge \textbf{n}_{d-k} \ge (x_{k+1},\dots,x_d)$ it follows that if we set $y$ to be $(y_1,\dots,y_d)$ then $y \in \comp(Y)$ and $y \ge x$, thereby proving our claim.	
\end{proof}

\subsubsection*{Fifth step}

\begin{lemma}
	$$|\refl(X)|_{A_d} \le q(|X|_{\pow(\nn^d)})$$ where $q(x) = (x+1)^d$.
\end{lemma}

\begin{proof}
	Let $X \in \pow(\nn^d)$. To bound the norm of $\refl(X)$, it is enough to bound the norm of every set of the form $X_{i_1,i_2,\dots,i_k}$. For ease of notation, we will only show that the norm of $X_{1,\dots,k}$ can be bounded. (The proof of the general case is just a syntactic modification of the following proof).
	Notice that to bound the norm of $X_{1,\dots,k}$ it is enough to bound the norm of every element in $X_{1,\dots,k}$ and the cardinality of $X_{1,\dots,k}$.
	
	Let $(x_1,\dots,x_k) \in X_{1,\dots,k}$. To bound the norm of $(x_1,\dots,x_k)$ it is enough to bound each $x_i$. For this purpose, we claim that
	\begin{equation} \label{eqn:boundclaim}
	\forall i, \ x_i \le |X|_{\pow(\nn^d)}
	\end{equation}
	
	Suppose equation (\ref{eqn:boundclaim}) is false and there exists $i$ such that $x_i > |X|_{\pow(\nn^d)}$. Without loss of generality assume that $i = k$. We consider two cases:\\
	
	\emph{Case 1: } $k = 1$: By equation (\ref{eqn:extensive}) we have that $\downarrow \comp(X) \neq \nn^d$ and so there exists $n \in \nn$ such that 
	\begin{equation} \label{eqn:compclaim}
	\textbf{n}_d \notin \downarrow \comp(X)
	\end{equation}
	
	We now have the following implications:
	
	\begin{align*}
	x_1 \in X_1 &\implies \downarrow\proj^{\comp(X)}_{1}(x_1) = \nn^{d-1}\\
	&\implies \exists (x_2,\dots,x_d) \in \proj^{\comp(X)}_1(x_1)  \\&\text{ such that } (x_2,\dots,x_d) \ge \textbf{n}_{d-1}\\
	&\implies (x_1,x_2,\dots,x_d) \in \comp(X)\\
	&\implies (x_1,\dots,x_d) \in \nn^d \setminus \uparrow X\\
	&\implies \forall (y_1,\dots,y_d) \in X, (y_1,\dots,y_d) \nleq (x_1,\dots,x_d)	
	\end{align*}

	Now since $x_1 > |X|_{\pow(\nn^d)}$ it has to be the case that 
	for every $(y_1,\dots,y_d) \in X, \ y_1 < x_1 < x_1+n+1$. This combined with the last implication gives us that 
	$$\forall (y_1,\dots,y_d) \in X, (y_1,y_2,\dots,y_d) \nleq (x_1+n+1,x_2,\dots,x_d)$$ 
	and so $(x_1+n+1,x_2,\dots,x_d) \in \comp(X)$. Since $(x_1+n+1,x_2,\dots,x_d) \ge \textbf{n}_d$, equation (\ref{eqn:compclaim}) is contradicted.\\

	\emph{Case 2: } $k > 1$. The proof is very similar to the first case. We have the following:
	\begin{align*}
	(x_1,\dots,x_k) \in X_{1,\dots,k} &\implies \downarrow \proj^{\comp(X)}_{1,\dots,k-1}(x_1,\dots,x_{k-1}) \neq \nn^{d-k+1}\\
	&\implies \exists \textbf{n}_{d-k+1} \text{ such that }\\& \hspace{0.75cm} \textbf{n}_{d-k+1} \notin \downarrow \proj^{\comp(X)}_{1,\dots,k-1}(x_1,\dots,x_{k-1})	
	\end{align*}

	Hence there exists $n \in \nn$ such that
	\begin{equation} \label{eqn:projclaim}
	\textbf{n}_{d-k+1} \notin \downarrow \proj^{\comp(X)}_{1,\dots,k-1}(x_1,\dots,x_{k-1}) 
	\end{equation}
	
	We now have the following implications:
	\begin{align*}
	(x_1,\dots,x_k) \in X_{1,\dots,k} &\\ 
	&\hspace{-1cm}\implies \downarrow\proj^{\comp(X)}_{1,\dots,k}(x_1,\dots,x_k) = \nn^{d-k}\\
	&\hspace{-1cm}\implies \exists (x_{k+1},\dots,x_d) \in \proj^{\comp(X)}_{1,\dots,k}(x_1,\dots,x_k) \\& \text{ such that } (x_{k+1},\dots,x_d) \ge \textbf{n}_{d-k}\\
	&\hspace{-1cm} \implies (x_1,\dots,x_k,x_{k+1},\dots,x_d) \in \comp(X)\\
	&\hspace{-1cm} \implies (x_1,\dots,x_d) \in \nn^d \setminus \uparrow X\\
	&\hspace{-1cm} \implies \forall (y_1,\dots,y_d) \in X, \\& (y_1,\dots,y_d) \nleq (x_1,\dots,x_d)	
	\end{align*}

	Now since $x_k > |X|_{\pow(\nn^d)}$, it has to be the case that 
	for all $(y_1,\dots,y_d) \in X, \ y_k < x_k < x_k+n+1$. This combined with the previous implication gives us that 
	\begin{multline*}
		\forall (y_1,\dots,y_d) \in X, (y_1,y_2,\dots,y_d) \nleq (x_1,\dots,x_{k-1},\\x_k+n+1,x_{k+1},\dots,x_d)	
	\end{multline*}
	and so $(x_1,\dots,x_{k-1},x_k+n+1,x_{k+1},\dots,x_d) \in \comp(X)$ which in turn implies that $(x_k+n+1,x_{k+1},\dots,x_d) \in$\\ $\proj^{\comp(X)}_{1,\dots,k-1}(x_1,\dots,x_{k-1})$. Since $(x_k+n+1,x_{k+1},\dots,x_d) \ge \textbf{n}_{d-k+1}$, equation (\ref{eqn:projclaim}) is contradicted.

\end{proof}

\section{Results and proofs for fast-growing hierarchies} \label{section:appcomp}

\subsection{Useful results for fast-growing hierarchies}

We present a collection of useful results for fast-growing classes before we proceed to prove the main theorems for majoring and minoring ordering.

Let $S(x)$ be the successor function.
Recall that $\{S_{\alpha}\}, \{S^{\alpha}\}, \{F_{\alpha}\}$ denote the Hardy, Cichon and fast-growing hierarchies for the successor function respectively. Notice that
\begin{equation} \label{eqn:HardyCichon}
	\forall x \in \nn, \ \forall \alpha < \epsilon_0, \ S_{\alpha}(x) = S^{\alpha}(x) - x
\end{equation}

\begin{prop} \label{prop:altcharac} (see Sections 2.2 and 5.3.1 of \cite{BeyondElem}):
	The class $\fastf_{\alpha}$ is the class of functions obtained by closure under substitution and limited recursion of the constant, sum, projections and the function $F_{\alpha}$. Each class $\fastf_{\alpha}$ is closed under (finite) composition. 
	Also for every $0 < \alpha < \beta < \epsilon_0$, we have $\fastf_{\alpha} \subseteq \fastf_{\beta}$.
	Further $\fastf_2$ is the 
	set of all elementary functions and 
	$\bigcup_{\alpha < \omega} \fastf_{\alpha}$ is the set of all primitive recursive functions.
	
\end{prop}

For the rest of this section, let $g$ be a fixed strictly increasing inflationary control function such that $g(x) \ge S(x)$. The following lemma is easy to see by induction on ordinals:
\begin{lemma} \label{lemma:indfastgrowing}
	For all $\alpha < \epsilon_0$ and for all $x \in \nn, \ f_{g,\alpha}(x) \ge F_{\alpha}(x), \ g_{\alpha}(x) \ge S_{\alpha}(x), \ g^{\alpha}(x) \ge S^{\alpha}(x)$.
\end{lemma}

The following facts are known about the fast growing hierarchy and the fast growing function classes (see section 2.3.3 of \cite{BeyondElem}, Lemmas C.12 and C.15 of \cite{ICALP} respectively):

\begin{prop} \label{prop:eventualdom}
	For any $0 < \alpha < \beta < \epsilon_0$, $F_{\beta} \in \fastf_{\beta}$ and $F_{\beta} \notin \fastf_{\alpha}$. Further if $h$ is a function such that \emph{eventually} $h(x) \ge F_{\beta}(x)$, then $h \notin \fastf_{\alpha}$ as well. (Consequently $F^c_{\beta} \notin \fastf_{\alpha}$ as well for any constant $c$).
\end{prop}


\begin{theorem} \label{theorem:impoupperbound}
	Let $g$ be eventually bounded by a function in $\fastf_{\gamma}$ where $\gamma > 0$. Then,
	\begin{itemize}
		\item If $\alpha < \omega$ then $f_{g,\alpha}$ is bounded by a function in $\fastf_{\gamma + \alpha}$ and 
		\item If $\gamma < \omega$ and $\alpha \ge \omega$ then $f_{g,\alpha}$ is bounded by a function in $\fastf_{\alpha}$.
	\end{itemize}
\end{theorem}

Using these facts we now prove the required theorems.

\subsection{Proof of theorem \ref{theorem:complowboundmaj}} \label{proof:complowboundmaj}

Fix a $d > 1$.
Let $\alpha = \omega^{\omega^{d-1}}$ and let $n$ be a sufficiently large natural number. By lemma \ref{lemma:indfastgrowing} we have that $H_{\alpha}(n) \le g_{\alpha}(n)$. By equation (\ref{eqn:HardyCichon}), $H_{\alpha}(n) = H^{\alpha}(n) - n$ and by proposition \ref{prop:funfactshier} we have that $H^{\alpha}(n) = F_{\omega^{d-1}}(n)$. Therefore $F_{\omega^{d-1}}(n) - n \le g_{\alpha}(n)$. 
By theorem \ref{theorem:lowboundmaj} we have that $g_{\alpha}(n) \le L_{(\pow(\nn^d), \maj),\varphi \circ g}(\varphi(n))$ which proves the first part.

For the second part, suppose $L_{(\pow(\nn^d), \maj), \varphi \circ g} \in \fastf_{\alpha}$ for some $\alpha < \omega^{d-1}$. 
Since $d > 1$, 
without loss of generality, we can let $2 \le \alpha < \omega^{d-1}$. 
Proposition~\ref{prop:altcharac} then implies 
that \\
$L_{(\pow(\nn^d), \maj), \varphi \circ g} (\varphi(n)) + n \in \fastf_{\alpha}$.
Since $F_{\omega^{d-1}}(n) \le$ \\
$L_{(\pow(\nn^d), \maj), \varphi \circ g}(\varphi(n)) + n$, proposition \ref{prop:eventualdom} now implies a contradiction.

\subsection{Proof of theorem \ref{theorem:compuppboundmaj}} \label{proof:compuppboundmaj}

Let $h(x) = 4x \cdot g(x)$ and let $\alpha=\omega^{(\omega^{d-1}) \cdot k}$. Since $g$ is primitive recursive, so is $h$ and so by proposition \ref{prop:altcharac}, $h \in \fastf_{\gamma}$ for some $\gamma < \omega$. Notice that $o(\pow(\nn^d)^k)$ is $\alpha$ and $\alpha$ is $4dk$-lean.

Let $n$ be a sufficiently large number. Hence, by theorem \ref{theorem:uppboundmaj} we have that $L_{(\pow(\nn^d)^k,\maj),g}(n) \le h_{\alpha}(4dkn)$. Now proposition \ref{prop:funfactshier} implies that $h_{\alpha}(4dkn) \le h^{\alpha}(4dkn) = f_{h,(\omega^{d-1}) \cdot k}(4dkn)$. Hence $L_{(\pow(\nn^d)^k,\maj),g}(n) \le f_{h,(\omega^{d-1}) \cdot k}(4dkn)$. Now applying  theorem \ref{theorem:impoupperbound} (and proposition \ref{prop:altcharac}) gives us the required upper bound.

\subsection{Proof of theorem \ref{theorem:complowboundmin}}
\label{proof:complowboundmin}

By theorem \ref{theorem:lowboundmin} we have that
\begin{equation} \label{eqn:majtomin}
	L_{(\pow(\nn^d),\maj),\varphi \circ g}(\varphi(n)) \le L_{(\pow(\nn^d),\smy),(p \circ \varphi \circ g)}(p(\varphi(n)))	
\end{equation}

By theorem \ref{theorem:complowboundmaj} we have that
\begin{equation} \label{eqn:fasttomaj}
	F_{\omega^{d-1}}(n) - n \le L_{(\pow(\nn^d), \maj),\varphi \circ g}(\varphi(n))
\end{equation}

Combining these two equations, we get the first part of the theorem.
By the same argument used in the proof of theorem~\ref{theorem:complowboundmaj},
we can show that if $L_{(\pow(\nn^d), \smy),p \circ \varphi \circ g} \in \fastf_{\alpha}$ for some $\alpha < \omega^{d-1}$
then $F_{\omega^{d-1}} \in \fastf_{\alpha}$, thereby contradicting 
proposition \ref{prop:eventualdom}.

\subsection{Proof of theorem \ref{theorem:compuppboundmin}}
\label{proof:compuppboundmin}

Let $A = \pow(\nn^d)^k$ and let $\alpha = \omega^{\omega^{d-1} \cdot (2^d \cdot k)}$. Let $t(x) := 4x \cdot q(g(x))$. Notice that $t$ is primitive recursive and so $t \in \fastf_{\gamma}$ for some $\gamma < \omega$. Now for sufficiently large $n$, there exists a constant $c$ such that
\begin{align*}
L_{(A,\lesmy_{A}),g}(n) &\le t_{\alpha}(c \cdot g(n)^{2d}) & \text{ by theorem \ref{theorem:uppboundmin} }\\
&\le t^{\alpha}(c \cdot g(n)^{2d}) & \text{ by proposition \ref{prop:funfactshier}}\\
&=f_{t,\omega^{d-1} \cdot (2^d \cdot k)}(c \cdot g(n)^{2d}) & \text{ by proposition \ref{prop:funfactshier}}
\end{align*}

Combining theorem \ref{theorem:impoupperbound} and proposition \ref{prop:altcharac}, we get the required result.

\end{document}